\documentclass[12pt]{article} 

\usepackage{url} 

\usepackage{amssymb}
\usepackage[utf8]{inputenc}
\usepackage[most]{tcolorbox}
\usepackage[skip=2pt, font=small]{caption}

\usepackage[colorlinks=true,linkcolor=blue]{hyperref}
\usepackage[ruled,vlined]{algorithm2e}
\usepackage[english]{babel}
\usepackage[title]{appendix}
\usepackage[round]{natbib}

\usepackage{enumerate}
\usepackage{mymacros}
\usepackage{graphics}
\usepackage{tabularx}
\usepackage{booktabs}
\usepackage{mathrsfs}
\usepackage{bm}
\usepackage{amsmath}
\usepackage{nccmath}
\usepackage{dsfont}
\usepackage{media9}
\usepackage{graphicx}
\usepackage{MnSymbol,wasysym}
\usepackage{subcaption}
\usepackage{textcomp}
\usepackage{wrapfig}
\usepackage{enumitem}
\usepackage{fancyhdr}
\usepackage{xcolor}
\usepackage{float}
\usepackage{makecell}
\usepackage{amsthm}
\usepackage{soul}
\usepackage{multirow}

\usepackage{newunicodechar}
\newunicodechar{ℓ}{\ensuremath{\ell}}

\newtheorem{theorem}{Theorem}[section]

                    \newcommand{\blind}{1}      

\begin{document}

\bibliographystyle{abbrvnat}

\def\spacingset#1{\renewcommand{\baselinestretch}%
{#1}\small\normalsize} \spacingset{1}


\if1\blind
{
  \title{\bf Using prior information to boost power in correlation structure support recovery}

  \author{Ziyang Ding \\
    Department of Statistics, Stanford University \\
    David Dunson \\
    Department of Statistical Science, Duke University}
  \maketitle
} \fi

\if0\blind
{
  \bigskip
  \bigskip
  \bigskip
  \bigskip
  \bigskip
  \bigskip
  \begin{center}
    {\LARGE\bf Using prior information to boost power in correlation structure support recovery}
\end{center}
  \medskip
} \fi

\bigskip
\begin{abstract}
Hypothesis testing of structure in correlation and covariance matrices is of broad interest in many application areas. In high dimensions and/or small to moderate sample sizes, high error rates in testing is a substantial concern.  This article focuses on increasing power through a frequentist assisted by Bayes (FAB) procedure.  This FAB approach boosts power by including prior information on the correlation parameters.  In particular, we suppose there is one of two sources of prior information: (i) a prior dataset that is distinct from the current data but related enough that it may contain valuable information about the correlation structure in the current data; and (ii) knowledge about a tendency for the correlations in different parameters to be similar, so that it is appropriate to consider a hierarchical model. When the prior information is relevant, the proposed FAB approach can have significant gains in power. A divide-and-conquer algorithm is developed to reduce computational complexity in massive testing dimensions. We show improvements in power for detecting correlated gene pairs in genomic studies, while maintaining control of Type I error or false discover rate (FDR). 
\end{abstract}

\noindent%
\vfill

\newpage
\spacingset{1.9}

\section{Introduction}

In this article, we consider the problem of inferring the locations of zeros in a
 $q \times q$ correlation matrix $\mathbf{P}$ given a sample of $n$ independent $q$-dimensional normal random vectors $\bx_1, \dots, \bx_n$, when prior information about entries of $\mathbf{P}$ is available. This problem can be expressed statistically as testing
\begin{eqnarray}
    H_{0wv}: \rho_{wv} = 0\quad \mbox{versus}\quad    
    H_{1wv}: \rho_{wv} \neq 0, \label{eq:hyp}
\end{eqnarray}
for each of the $q(q-1)/2$ unique entries of $\mathbf{P}$, with $\rho_{wv}$ denoting the correlation between the $w^{th}$ and $v^{th}$ variables. This testing problem is commonly known as \textit{support recovery} of a correlation structure, or correlation structure testing, which is a fundamental problem in multivariate analysis. There are important applications in finance \citep{chen2020testing}, biology \citep{van2019genetic}, and many other fields. 

This article proposes a methodology for improving power in frequentist support recovery by exploiting prior information.  This prior information takes one of two forms: (i) a prior dataset is available that contains information on the correlation parameters of interest in the current dataset, but with the prior dataset distinct enough so that it does not make sense to merge the datasets; and (ii) knowledge that the correlation parameters for different pairs of variables tend to be related, so that it is reasonable to borrow information. In case (i) for simplicity in exposition we assume that the external dataset $\mathbf{X}_{\text{ext}}$ has the same format as $\mathbf{X}$, though modifications to account for the case in which $\mathbf{X}_{\text{ext}}$ measures a distinct but overlapping set of variables are straightforward.  In case (ii) we assume that prior information relevant for testing $H_{0wv}$ vs $H_{1wv}$ takes the form of test statistics for other pairs of variables.  


Finding zeros in the correlation matrix is mainly achieved via one of two general approaches. The first is estimation-based, using a sparseness favoring correlation matrix estimator; for example, 
via $\ell_1$-penalized maximum likelihood. Related works include \cite{friedman2008sparse, bickel2008regularized, lam2009sparsistency, rothman2009generalized, cai2011constrained, fan2016overview} and \citet{cai2016estimating}. The second general approach is testing-based, in which we form hypotheses on individual correlation coefficients in the matrix being zero. Examples of this approach include work by \citet{yuan2006model, cai2011constrained, cai2013two, cai2016inference, kundu2019efficient} and \citet{lee2021maximum}. Refer to \citet{cai2017global} and \citet{na2021estimating} for an overview of related methods. This article focuses on the testing-based approach. We aim to improve the accuracy of correlation matrix support recovery through the use of prior information, either through external data or borrowing information from a current dataset.
Relevant external data may be from a previous related study that provides (potentially imperfect) information on the locations of zeros in ${\bf P}$.

Apart from research on correlation and covariance structure testing, there is also a previous related literature that seeks to improve power for hypothesis testing in general, either by exploiting prior information of various types or introducing auxiliary information. One popular approach that helps gain power for adaptive FDR-controlling procedures by estimating the overall proportion of true nulls is derived by \citet{storey2002direct}. Another general category of approaches include prior weights in the testing procedure, with these prior weights providing information on hypotheses that should be more or less likely {\em a priori}. \citet{genovese2006false} developed a weighted variation of the Benjamini-Hochberg (BH) procedure \citep{Benjamini1995}. Related approaches can be found in \citet{dobriban2015optimal} and \citet{dobriban2016general}. Other work has included prior information through other mechanisms; for example, by assuming the side information is independent of the $p$-values \citep{roquain2009optimal}, or by allowing for data-dependent weighting \citep{hu2010false, zhao2014weighted, ignatiadis2016data, durand2019adaptive, lei2018adapt, ignatiadis2021covariate, luo2025sequda}. Another idea is to 
exploit a prior ordering to concentrate power on more ``promising'' hypotheses near the top of the ordering \citep{Barber2015, g2016sequential, li2017accumulation, lei2016power}. The above methods provide a rich toolbox through improving power with the introduction of certain types of external information.  However, the focus has been largely on including direct prior information about the different hypotheses under consideration, with a set of $p$-values provided as input and not informed by this prior information - what is informed is the decision rule based on these $p$-values.  

In our work, focusing on the correlation support recovery problem, we take the different approach of including prior information relevant to the correlations under consideration, instead of the hypotheses directly, with this prior information used in obtaining the $p$-values.  The proposed method is based on the frequentist assisted by Bayes (FAB) procedure of \cite{hoff2021smaller}.  To simplify the methodology, we rely on Fisher-transformations of the correlation coefficients and associated asymptotic approximations.  Prior information is incorporated via a probabilistic model linking different hypothesis tests.  The resulting $p$-values can be used to maintain Type I error rates and FDR control.  When the prior information is appropriate, power can be improved relative to tests that ignore this information.  To reduce computational complexity and facilitate implementations in high dimensions, we introduce asymptotic approximations and a divide-and-conquer strategy, while providing asymptotic theory supporting these approaches.
    We provide examples incorporating external information of different types for FAB testing. For instance, if the entries of $\mathbf{P}$ are the population-level correlations among a set of $q$ genes, then knowledge about the positive regulation between a pair of genes in a well-studied biological pathway can constitute external information about $\mathbf{P}$. Alternatively, external information can take the form of a posited low-rank structure for $\mathbf{P}$, which determines groups of entries whose signs and magnitudes are similar.
    

\section{Preliminaries}

In this section, we briefly review some past techniques needed to construct our FAB correlation structure test. Section \ref{sec: fisher z-transform} reviews the well known Fisher variance stabilizing function under the context of correlation testing. Section \ref{sec: FAB testing} reviews the ``frequentist assisted by Bayes'' (FAB) approach developed by \citet{hoff2021smaller}, the methodology that we extend into the correlation structure testing setting. 

\subsection{Fisher-transformed correlation coefficient} \label{sec: fisher z-transform}

Our focus is on testing the hypotheses in (\ref{eq:hyp}).  Denote the sample Pearson correlation coefficient as
\begin{equation}
    \hat{\rho}_{wv} = \frac{\sum_{i=1}^{n}(x_{iw} - \bar{x}_w)(x_{iv} - \bar{x}_v)}{\sqrt{\sum_{i=1}^{n}(x_{iw} - \bar{x}_w)^2 \sum_{i=1}^{n}(x_{vk} - \bar{x}_v)^2}}
\end{equation}
for each of the $q(q-1)/2$ unique entries of $\mathbf{P}$. Under the null hypothesis $H_{0wv}$, the statistic
\begin{equation}
    t_{\rho_{wv}}(\bX) = \hat{\rho}_{wv}\sqrt{(n-2) /\left(1-\hat{\rho}_{wv}^{2}\right)}
\end{equation}
has an asymptotic $t$-distribution with $n-2$ degrees of freedom, so a level-$\alpha$ test can be constructed by rejecting $H_{0wv}$ when $|t_{\rho_{wv}}(\bX)|$ exceeds the $1 - \alpha/2$ quantile of the $t_{n-2}$ distribution. 
Typically, however, an approximately level-$\alpha$ test is constructed using Fisher's transformed test statistic \citep{fisher1915frequency}
\begin{equation} \label{eq:fisher equation}
    T_{\rho_{wv}}(\bX) = F(\hat{\rho}_{wv}) = \frac{1}{2} \log \frac{1+\hat{\rho}_{wv}}{1-\hat{\rho}_{wv}}
\end{equation}
which is approximately distributed as
\begin{equation}
    T_{\rho_{wv}}(\bX) \sim N\left(F(\rho_{wv}), \frac{1}{n-3}\right)
\end{equation}
Because its variance is approximately constant as a function of $\rho_{wv}$, the statistic $T_{\rho_{wv}}(\bX)$ can be used to construct more flexible and stable hypothesis tests than those based on $t_{\rho_{wv}}(\bX)$ or using the exact sampling distribution of $\hat{\rho}_{wv}$. The approximate normality of $T_{\rho_{wv}}(\bX)$ is useful to our strategy for incorporating external information into testing the dependence structure of $\mathbf{P}$.

\subsection{FAB Testing} \label{sec: FAB testing}

Given an iid sample $x_i \sim N(\theta, \sigma^2),~i=1, \dots, n$, suppose interest focuses on testing 
\[
    H_0 : \theta = 0\quad \mbox{versus}\quad H_1 : \theta \neq 0
\]
Supposing $\sigma$ is known, the $p$-value for the classical uniformly most powerful unbiased (UMPU) test has closed form
\[
p^{\text{UMPU}} = 1 - |\Phi( \overline{x} / \sigma) - \Phi( -\overline{x} / \sigma)|.
\]
Under the null hypothesis $H_0$, $p^{\mathrm{UMPU}}$ is uniformly distributed in the interval $(0, 1)$, so the Type I error rate of the UMPU test may be controlled at level $\alpha$ by rejecting the null when $p^{\mathrm{UMPU}} < \alpha$. The FAB test introduced by \cite{hoff2021smaller} leverages the fact that for any offset random variable $b$ that is independent of $\overline{x} / \sigma$, the  ``FAB $p$-value"
\begin{equation}
\label{eq:FAB_b}
\begin{aligned}
p^{\mathrm{FAB}} := 1 - |\Phi( \overline{x} / \sigma + b) - \Phi( -\overline{x} / \sigma)|
\end{aligned}
\end{equation}
is also uniformly distributed under the null hypothesis $H_0$. 
\begin{theorem} \label{thm:independent_b}
\citep{hoff2021smaller} Let $Z$ and $b$ be independent random variables with $Z \sim N(0,1)$. Then $\operatorname{Pr}(1-
|\Phi(Z+b)-\Phi(-Z)|<u)=u$
\end{theorem}
Rejecting the null when $p^{\mathrm{FAB}} < \alpha$ is therefore also an $\alpha$-level frequentist test. However, the power of the FAB test may be greater than that of the UMPU test if $b$ is chosen carefully. We 
use external information to inform our choice of $b$ for correlation structure hypothesis testing. We reserve the explanation of our specific approach to incorporating external information for Section \ref{sec:Methods}, and highlight the most important aspects of the general FAB testing approach here:
\begin{enumerate}
    \item For the case of testing several hypotheses, a probability model---called a ``linking model"---is proposed. Combined with the sampling model for the test statistics, the linking model can be used to perform FAB tests for each parameter of interest; in our case, these are the $\rho_{jk}$.
    \item The offset term $b_{jk}$ for each hypothesis test must be inferred from information independent from the test statistics. It could be inferred either completely from an external dataset or from independent information within the same dataset. Details regarding the two situations will be discussed in Section \ref{sec: independent case} and Section \ref{sec: nonindependent case}.
    \item The design of the linking model impacts the power of FAB tests, but it does not affect the Type I error when $b_{jk}$ is independent of the test statistic.
\end{enumerate}


\begin{table}
\footnotesize
\begin{tabularx}{\textwidth}{p{0.12\textwidth}X}
\toprule
\textbf{Notation}               & \textbf{Definition} \\
\midrule
\textbf{Function} \\
$F(\cdot)$                      & Fisher's z-transformation. Defined in Equation \ref{eq:fisher equation}\\

\textbf{Variables}\\
$\bm{\Rho}, \bm{\rho}, \rho$                        & Matrix, vector, and scalar valued true correlation\\
$\mathbf{Z}, \bm{z}, z$                                 & Matrix, vector, and scalar valued Fisher-transformed true correlation\\
$\bm{\Omega_z}, \bm{\omega_z}, \omega_z$            & Matrix, vector, and scalar valued \textbf{true correlation of Fisher-transformed correlation estimators}: $\mathrm{Cor} (\widehat{z}_1, \widehat{z}_2)$\\
$m_j, v_j$                                          & FAB offset mean and variance of the $j^{th}$ test\\
$\bG_j$                                             & Decorrelation matrix of the $j^{th}$ test\\

\textbf{Estimator}\\
$\widehat{(\cdot)}$                                     & Estimator constructed from direct test statistics.\\
$\widetilde{(\cdot)}$                                   & Estimator constructed from indirect information conditionally independent from test statistics.\\

\textbf{Indexing}\\
$i \in \{1:n\}$                 & $i$ is the index of subjects $\{X_i\}_{i \in \{1:n\}} := \{X_1, X_2, \cdots X_n \}$\\
$j \in \{1:p\}$                 & $j$ is the index of test statistics $\{\widehat{z}_j\}_{j \in \{1:p\}} := \{\widehat{z}_1, \widehat{z}_2, \cdots \widehat{z}_p \}$\\
$k \in \{1:m\}$                 & $k$ is the index of hypothesis group $\{\widehat{\bm{z}}^{(k)}\}_{k \in \{1:m\}} := \{\widehat{\bm{z}}^{(1)}, \widehat{\bm{z}}^{(2)}, \cdots \widehat{\bm{z}}^{(m)} \}$\\
$n$                             & Total number of samples\\
$q$                             & Total number of dimensions of subject $\bX$\\
$p$                             & Total number of tests. $p = q$ if hypothesis tests are on sample mean. But $p = q(q-1)/2$ if tests are on pairwise correlation.\\
$m$                             & The total number hypothesis groups. Formal description introduced in Section \ref{sec: FAB corrStruct testing}.\\
$-j$                            & A negative sign before an index indicates every other index instead of this index. For example, $\widehat{\bm{z}}_{-j} := \{\widehat{z}_{1}, \cdots \widehat{z}_{j-1}, \widehat{z}_{j+1} \cdots , \widehat{z}_p\}$.\\

\bottomrule
\end{tabularx}
\vspace{10pt}
\caption{\small Notation of frequently used parameters} \label{table: notation}
\label{table:notation}
\end{table}
\normalsize

\section{Methodology} \label{sec:Methods}

In this section, we describe our methodology for FAB correlation structure testing. 
The commonly used notation in the paper is listed in Table \ref{table: notation}.
In Section \ref{sec: FAB cor testing} we adapt the methods in Section \ref{sec: FAB testing} to the setting of correlation coefficient testing. This section focuses on forming the FAB test on a single correlation coefficient given external information. In Section \ref{sec: FAB corrStruct testing} we address the computational challenge of performing $O(q^2)$ FAB tests for correlation structure by adopting a divide-and-conquer approach.

\subsection{FAB Correlation Testing} \label{sec: FAB cor testing}

As introduced in Section \ref{sec: fisher z-transform}, the distribution of the Fisher-transformed sample correlation coefficient $\widehat{z}_{wv}$ approximately follows the normal distribution
\begin{equation} \label{eq:summary_stat_dist}
\begin{aligned}
\widehat{z}_{wv} := F(\hat{\rho}_{wv}) \sim \cN \left( F(\rho_{wv}), \frac{1}{n-3} \right)
\end{aligned}
\end{equation}
The approximate normality of the statistic $z_j$ allows us to adapt the FAB methodology for testing described in Section \ref{sec: FAB testing}. For notational simplicity in what follows, we simplify the indexing by using a single index $j=1,\ldots,p=q(q-1)/2$ for the different $w,v$ pairs.


Suppose we are testing the $j^{th}$ correlation $\rho_j$ based on its observed Fisher transformed correlation coefficient $\widehat{z}_j$. The FAB correlation test improves power of the test on $\widehat{z}_j$ by borrowing information from some other statistics $\widehat{\bm{z}}_{j}^\prime$, which can be statistics from a purely external dataset $\widehat{\bm{z}}_{-j}^{\text{ext}}$ (take $\widehat{\bm{z}}_{j}^\prime = \widehat{\bm{z}}_{-j}^{\text{ext}}$ as in Section \ref{sec: independent case}) or can be statistics from other tests $\widehat{\bm{z}}_{-j}$ within the same dataset as $\widehat{z}_j$ (take $\widehat{\bm{z}}_{j}^\prime = \widehat{\bm{z}}_{-j}$ as in Section \ref{sec: nonindependent case}). The test statistics $z_j$ and the indirect information source $\bm{z}_j^\prime$ are modeled jointly using the following model:
\begin{align}
    \widehat{\bm{z}} = 
    \begin{bmatrix}
    \widehat{z}_j \\
    \widehat{\bm{z}}_j^{\prime}
    \end{bmatrix}
    &\sim \cN \left( \bm{z} := 
    \begin{bmatrix}
    z_j\\ 
    \bm{z}_j^{\prime}
    \end{bmatrix}
    , \quad \frac{1}{n-3} \bm{\Omega_z} := \frac{1}{n-3}\mathrm{Cor}(\widehat{\bm{z}}) \right) \label{eq:FAB corr sampling model} \\
    \bm{z} &\sim \cN(\bm{\mu} := \mathbf{W}\bm{\eta}, \bm{\Psi}) \label{eq:FAB corr linking model}
\end{align}

Equation \ref{eq:FAB corr sampling model}, which specifies the distribution of the observed Fisher-transformed correlation coefficient $\widehat{z}_j := F(\widehat{\rho_j})$ and that of other observed statistics $\widehat{\bm{z}}_j^{\prime}$, is called the sampling model. Equation \ref{eq:FAB corr linking model}, which is a custom designed prior distribution on the true Fisher-transformed correlation $z_j$ that ``links'' the true parameter $z_j$ with $\bm{z}_j'$, is called the linking model. This linking model allows sharing of information between $z_j$ and $\bm{z}_j'$ with $\bm{\mu}$ having a lower dimensional representation $\bm{\eta}$ with $\mathrm{dim}(\bm{\eta}) < \mathrm{dim}(\bm{\mu})$. The lower dimensional structure is characterized by a factor model as $\bm{\mu} = \mathbf{W}\bm{\eta}$, in which the factor loading matrix $\mathbf{W}$ is pre-defined, whereas $\bm{\eta}$ is unknown and will be estimated using empirical Bayes. When testing is based on $\widehat{z}_j$, FAB borrows indirect information from $\widehat{\bm{z}}_j^{\prime}$ in the sampling model and shares it with $\widehat{z}_j$ via the linking model, therefore improving power. 

As described in \cite{hoff2021smaller}, we can form a FAB test that is approximately Bayes-optimal with respect to power using the following form for the offset terms $b_j$:
\begin{align}
    p^{\text{FAB}}_j &:= 1 - |\Phi( \widehat{z_j} \sqrt{n-3} + \widetilde{b_j}) - \Phi( -\widehat{z_j} \sqrt{n-3})| \label{eq:FAB corr p}\\
    \widetilde{b_j} := \frac{2\widetilde{m_j}}{\widetilde{v_j}\sqrt{n-3}} \quad 
    \widetilde{m_j} &:= \bE  \left[   z_j | \widetilde{\bm{\mu}}(\widetilde{z}_j^{\idpt}), \widetilde{\bm{\Psi}}(\widetilde{z}_j^{\idpt}), \widetilde{z}_j^{\idpt} \right] \quad
    \widetilde{v_j} := \bVar\left[   z_j | \widetilde{\bm{\mu}}(\widetilde{z}_j^{\idpt}), \widetilde{\bm{\Psi}}(\widetilde{z}_j^{\idpt}), \widetilde{z}_j^{\idpt} \right] \nonumber    
\end{align}
where $\widetilde{z}_j^{\idpt}$ is the indirect information, drawn from $\widehat{\bm{z}}$, that is independent from the test statistics $\widehat{z}_j$. $\widetilde{\bm{\mu}}, \widetilde{\bm{\Psi}}$ are estimators of $\bm{\mu}$, $\bm{\Psi}$ is estimated only from indirect information  $\widetilde{z}_j^{\idpt}$ to ensure independence with test statistics $\widehat{z}_j$ as required by Theorem \ref{thm:independent_b}, and $\widetilde{m}_j, \widetilde{v}_j$ are the posterior mean and variance of $z_j$ based on the empirical Bayes prior $\bm{z} \sim \cN(\widetilde{\bm{\mu}}, \widetilde{\bm{\Psi}})$ and indirect data likelihood $\mathcal{L}(\widetilde{z}_j^{\idpt}; \bm{z})$. Therefore, $\widetilde{\bm{\mu}}, \widetilde{\bm{\Psi}}, \widetilde{m}_j, \widetilde{v}_j$ are functions of only $\widetilde{z}_j^{\idpt}$, and $g: \widehat{\bm{z}}  \mapsto \widetilde{z}_j^{\idpt} $ is a ``make-independence'' function $g$ of $\widehat{\bm{z}}$ to ensure all the later calculated statistics, $\widetilde{\bm{\mu}}, \widetilde{\bm{\Psi}}, \widetilde{m}_j, \widetilde{v}_j$, are independent with $\widehat{z}_j$.


The construction of such $g$ functions depends on the application scenario. If $\widehat{\bm{z}}_{j}^\prime$ comes from an external data source (scenario when $\widehat{\bm{z}}_{j}^\prime = \widehat{\bm{z}}_{-j}^{\text{ext}}$), the correlation matrix $\bm{\Omega_z}$ of $\widehat{\bm{z}}$ is block-diagonal. Hence, the independent indirect information $\widehat{z}_{i}^{\idpt}$ can be easily sourced from $\widehat{\bm{z}}$ with a trivially constructed function $g$. However, if $\widehat{\bm{z}}_{j}^\prime$ is sourced from other test statistics within the same dataset as $\widehat{z}_j$ (scenario when $\widehat{\bm{z}}_{j}^\prime = \widehat{\bm{z}}_{-j}$), the correlation matrix $\bm{\Omega_z}$ of $\widehat{\bm{z}}$ may be unknown. Under such circumstances, construction of $g$ requires more effort.
 
We split the discussion according to the two different scenarios - external or internal information.  Section \ref{sec: independent case} introduces cases when $\widehat{\bm{z}}_{j}^\prime = \widehat{\bm{z}}_{-j}^{\text{ext}}$, so that $\widehat{z_j} \perp\!\!\!\perp \widehat{\bm{z}}_j^{\prime}$ innately. Section \ref{sec: nonindependent case} discusses when $\widehat{\bm{z}}_{j}^\prime = \widehat{\bm{z}}_{-j}$ is sourced from the same dataset so that independence is not automatic.
 
\subsubsection{\texorpdfstring{$\widehat{\bm{z}}_j^{\prime} = \widehat{\bm{z}}_{-j}^{\text{ext}}$}{Lg}: Sampling Model Using External Data } \label{sec: independent case}

This section focuses on the scenario when $\widehat{\bm{z}}_j^{\prime}$ is calculated using an external dataset, so that $\widehat{z_j} \perp\!\!\!\perp \widehat{\bm{z}}_j^{\prime}$ conditional on the true parameter. The method assumes the historical or external dataset shares a similar correlation structure as the current dataset.  For example, in studying dependence across genes in their expression levels using single cell RNAseq data, it is common to collect data under similar conditions.
For instance, in the current setting, we may be interested in identifying correlated pairs of genes for squamous epithelium cells, while borrowing information from data previously collected for cuboidal epithelium cells.
There is substantial similarity in mRNA expression profiles  
for related cell types \citep{van2003cells}.  Hence, 
 we treat the latter (cuboidal) dataset as the external dataset, from which $\widehat{\bm{z}}_j^{\prime} = \widehat{\bm{z}}_{-j}^{\text{ext}}$ are calculated. As $\widehat{\bm{z}}_j^{\prime}$ is from a different dataset that is independent from $\widehat{z_{j}}$, the required independence of $\widetilde{m_j}, \widetilde{v_j}$ is satisfied. The following sampling model summarizes the distribution of $(\widehat{z_{j}}, \widehat{\bm{z}}_j^{\prime} )^\top$
\begin{align}
    \widehat{\bm{z}} = 
    \begin{bmatrix}
        \widehat{z}_j \\
        \widehat{\bm{z}}_j^{\prime} := \widehat{\bm{z}}_{-j}^{\text{ext}}
    \end{bmatrix}
    &\sim \cN \left(
    \bm{z} = 
    \begin{bmatrix}
        z_j\\
        \bm{z}_j^{\prime} := \bm{z}_{-j}^{\text{ext}}
    \end{bmatrix}
    , 
    \frac{1}{n-3} 
    {\bm{\Omega_{z}}}
    =
    \frac{1}{n-3} 
    \begin{bmatrix}
        1 & \bm{0}\\
        \bm{0} & \bm{\Omega}_{\bm{z}_j^\prime}
    \end{bmatrix}
    \right) \label{eq: temp arrangement}
\end{align}

The function $g$ can be simply constructed as a decorrelation matrix $\bG_{j}$, a $p \times (p-1)$ dimensional matrix obtained by deleting the column corresponding to $\widehat{z}_j$ from $\mathbf{I}_p$. In the arrangement of Equation \ref{eq: temp arrangement}, as $\widehat{z}_j$ is in the first dimension of $\widehat{\bm{z}}$, the first column of $\mathbf{I}_p$ is deleted to construct $\mathbf{G}_j$

\[
\mathbf{G}_j = 
\begin{bmatrix}
    \bm{0}_{1 \times (p-1)}\\
    \mathbf{I}_{(p-1)\times (p-1)}
\end{bmatrix}
\]

As a result, the indirect information $\bG_{j}^\top \widehat{\bm{z}} = \widetilde{z}_j^{\idpt}$ is independent from $\widehat{z}_j$. Thus, we can obtain an empirical Bayes estimator for $\widetilde{\bm{\eta}}, \widetilde{\bm{\Psi}}$ solely from the marginal distribution of the indirect information $\widetilde{z}_j^{\idpt} = \mathbf{G}_{j}^{\top} \widehat{\bm{z}}$ as in Equation \ref{eq: marginal indirect distribution}. The estimator can be obtained via maximum likelihood or $l_2$ penalization, as long as it is independent from $\widehat{z}_j$.

\begin{align}
\mathbf{G}_{j}^{\top} \widehat{\bm{z}} &\sim \cN_{p-1}\left(\mathbf{G}_{j}^{\top} \mathbf{W}\bm{\eta}, \mathbf{G}_{j}^{\top} (\frac{1}{n-3}\bm{\Omega_z} + \bm{\Psi} ) \mathbf{G}_{j}\right) \label{eq: marginal indirect distribution}
\end{align}
The posterior of $\bm{z}$ given the indirect information $\widetilde{z}_j^{\idpt}$ and the empirical Bayes estimated prior parameters $\widetilde{\bm{\eta}}, \widetilde{\bm{\Psi}}$ can be easily obtained using linear algebra:
\begin{equation} \label{eq:FAB pos mean var}
    \begin{aligned}
    \bm{z} &| \mathbf{G}^{\top}_j \widehat{\bm{z}}, \widetilde{\bm{\eta}}, \widetilde{\bm{\Psi}} \sim \cN_p( \bm{m}, \mathbf{V}) \\
    \mathbf{V} &=\left[\widetilde{\bm{\Psi}}^{-1}+(n-3)\mathbf{G}_{j}\left(\mathbf{G}_{j}^{\top} \bm{\Omega_z} \mathbf{G}_{j}\right)^{-1} \mathbf{G}_{j}^{\top}\right]^{-1} \\
    \bm{m} &=\mathbf{V}\left[\widetilde{\bm{\Psi}}^{-1} \mathbf{W}\widetilde{\bm{\eta}}+(n-3)\mathbf{G}_{j}\left(\mathbf{G}_{j}^{\top} \bm{\Omega_z} \mathbf{G}_{j}\right)^{-1} \mathbf{G}_{j}^{\top} \widehat{\bm{z}}\right]
    \end{aligned}
\end{equation}

With $\bm{m}, \mathbf{V}$ calculated, the FAB $p$-value can then be simply calculated as in Equation \ref{eq:FAB corr p}, where $\widetilde{m}_j, \widetilde{v}_j$ are defined as the posterior mean and variance of $z_j$. In our arrangement of Equation \ref{eq: temp arrangement}, $\widetilde{m}_j, \widetilde{v}_j$ are the first row of $\bm{m}$ and the first row first column element of $\mathbf{V}$ respectively, as $z_j$ is in the first dimension of $\bm{z}$.

\subsubsection{ \texorpdfstring{$\widehat{\bm{z}}_j^{\prime} = \widehat{\bm{z}}_{-j}$}{Lg}: Sampling Model Using Internal Data  } \label{sec: nonindependent case}

In the absence of an external dataset containing a similar correlation structure, we can borrow information across the different Fisher-transformed correlation coefficients in the same dataset.  It is well known that borrowing across seemingly unrelated parameters can yield statistical dividends.  In the motivating gene expression application, it is plausible to suppose that the correlations between other pairs of genes are informative about the correlation between a particular pair of interest.  This motivates letting  $\widehat{\bm{z}}_j^{\prime} = \widehat{\bm{z}}_{-j}$ to draw indirect information from other hypotheses' test statistics $\widehat{\bm{z}}_{-j}$ to assist the testing on $z_j$. However, independence between $\widehat{\bm{z}}_j^{\prime}$ and $\widehat{z}_j$ does not hold any longer. Furthermore, the correlation matrix $\bm{\Omega_z}$ is completely unknown. Additional estimations and statistical procedures on $\bm{\Omega_z}$ are needed to ensure the validity of the FAB correlation test. The formal sampling model in this scenario is specified as 
\begin{align}
    \widehat{\bm{z}} = 
    \begin{bmatrix}
        \widehat{z}_j \\
        \widehat{\bm{z}}_j^{\prime} := \widehat{\bm{z}}_{-j}
    \end{bmatrix}
    &\sim \cN \left(
    \bm{z} = 
    \begin{bmatrix}
        z_j\\
        \bm{z}_j^{\prime} := \bm{z}_{-j}
    \end{bmatrix}
    , 
    \frac{1}{n-3} 
    {\bm{\Omega_{z}}}
    \right) \label{eq: arrangement nonindependent}
\end{align}


First, assuming that $\bm{\Omega_z}$ is known and $\widehat{z_j} \not{\perp\!\!\!\perp} \widehat{\bm{z}}_{-j}$, we can extract independent information by constructing $\mathbf{G}_j$ as a decorrelation matrix, whose columns form a basis of the null space of $\widehat{z}_j$'s corresponding column of $\bm{\Omega_z}$ (which is the first column of $\bm{\Omega_z}$ in the arrangement of Equation \ref{eq: arrangement nonindependent}). As $\widehat{\bm{z}}$ is normally distributed, the indirect information $\mathbf{G}_j^\top \widehat{\bm{z}} = \widetilde{z}_j^{\idpt} $ is  independent from $\widehat{z}_j$. However, in reality $\bm{\Omega_z}$ is unknown and requires estimation. In fact, we can construct a consistent estimator $\widehat{\bm{\Omega_z}}$ for $\bm{\Omega_z}$, from which we can also construct a consistent estimator $\widehat{\mathbf{G}}_j$ for the decorrelation matrix $\mathbf{G}_j$ due to the continuous mapping theorem. This consistent estimator of $\widehat{\mathbf{G}}_j$ suffices to provide asymptotic independence between $\widehat{z_j}$ and $\widehat{\mathbf{G}}_j^\top \widehat{\bm{z}}$.

\begin{theorem} \label{thm: consistency control}
Let $\widehat{\bm{\Omega_z}} \stackrel{pr.}{\longrightarrow} \bm{\Omega_z} $ and $\sqrt{n+3}(\widehat{\boldsymbol{z}}-\boldsymbol{z}) \stackrel{dist.}{\longrightarrow} \bm{r}$, where $\bm{r} \sim N_{p}(\mathbf{0}, \bm{\Omega_z}) .$ Then as $n \to \infty$
\begin{enumerate}
    \item $\operatorname{Cor}\left[\widehat{\bG}_{j}^{\top} \widehat{\boldsymbol{z}}, \widehat{z}_{j}\right] \to \mathbf{0}$
    \item $\operatorname{Pr}\left(\left\{\sqrt{n} \widehat{\bG}_{j}^{\top}(\widehat{\boldsymbol{z}}-\boldsymbol{z}) \in A\right\} \cap\left\{\sqrt{n}\left(\widehat{z}_{j}-z_{j}\right) \in B\right\}\right) \to \operatorname{Pr}\left(\bG_j^{\top} \bm{r} \in A\right) \times \operatorname{Pr}\left(r_{j} \in B\right)$
\end{enumerate}
where the columns of $\bG_{j}$ form a basis spanning the null space of the $\widehat{z}_j$-corresponding column of $\bm{\Omega_z}$; the columns of $\widehat{\bG}_{j}$ form a basis spanning the null space of the $\widehat{z}_j$-corresponding column of $\widehat{\bm{\Omega_z}}$; and $A \subset \mathbb{R}^{n-1}$ and $B \subset \mathbb{R}$ are measurable sets.
\end{theorem}

Proof of Theorem \ref{thm: consistency control} is in supplementary material. Hence, the challenge is in obtaining a consistent estimator of $\bm{\Omega_z}$, the correlation matrix of a set of Fisher-transformed correlation coefficients $\widehat{\bm{z}}$. As we have only one measurement of  $\widehat{\bm{z}}$, to consistently estimate $\bm{\Omega_z}$, we apply the bootstrap.

\begin{theorem} \label{thm:bootstrap consistency of big matrix}
Let $\widehat{\mathcal{Z}^*_B} := \{\widehat{\bm{z}}^*_1, \widehat{\bm{z}}^*_2, \cdots, \widehat{\bm{z}}^*_B\}$ be $B$ bootstrap samples. Define $\widehat{\bm{\Omega_z}}^*$ as the bootstrap estimator of $\bm{\Omega_z}$ based on the $B$ bootstrap samples $\widehat{\mathcal{Z}^*_B}$. Then,
$\widehat{\bm{\Omega_z}}^* \stackrel{pr.}{\longrightarrow} \bm{\Omega_z}.$
\end{theorem}




The proof of Theorem \ref{thm:bootstrap consistency of big matrix} is in supplementary material. 
Under such consistency, the distribution of each FAB $p$-value converges to a uniform($0,1$) distribution asymptotically. Hence, it is reasonable to use $\widehat{\bm{\Omega_z}}^*$ as $\bm{\Omega_z}$ to construct the decorrelation matrix $\widehat{\bG_j}$. The FAB $p$-value inference algorithm is summarized in Algorithm \ref{alg: FAB correlation test}.

{
\vspace{0.25in}
\begin{algorithm}[H] \label{alg: FAB correlation test}
\caption{General Bootstrap FAB correlation test}
\label{algo:Offline}
\SetAlgoLined
\vspace{0.15in}
\KwResult{ FAB $p$-values for $p = q(q-1)/2$ correlation tests }
\textbf{Input}: Fisher-transformed coefficient estimator $\widehat{\bm{z}}$ calculated from the $q$ dimensional data $\bX$ \\
\textbf{Sample}: $B$ bootstrap samples $\widehat{\mathcal{Z}^*_B} := \{ \widehat{\bm{z}}^*_1, \widehat{\bm{z}}^*_2, \cdots \widehat{\bm{z}}^*_B\}$ by re-sampling $n$ samples from $\bX$ with replacement for $B$ times.\\
\textbf{Calculate}: correlation matrix estimator $\widehat{\bm{\Omega_z}}$ for $\widehat{\bm{z}}$ from bootstrap samples $\widehat{\mathcal{Z}^*_B}$

\For{$j = 1,2,\cdots p$}
{
    \textbf{Construct:} $\widehat{\mathbf{G}}_j$ as the matrix whose columns form a basis for the null space of the $\widehat{z}_j$ corresponding column of $\widehat{\bm{\Omega_z}}$\\
    \textbf{Calculate:} Empirical Bayes estimator for $\widetilde{\bm{\eta}}, \widetilde{\bm{\Psi}}$ from indirect information $\widetilde{z}_j^\idpt = \widehat{\mathbf{G}}_j^\top \widehat{\bm{z}}$\\
    \textbf{Calculate:} Conditional mean and variance of $z_j | \widetilde{\bm{\eta}}, \widetilde{\bm{\Psi}}, \widetilde{\bm{z}}^{\idpt}$ as in Equation \ref{eq:FAB pos mean var}, denoted as $\widetilde{m}_j, \widetilde{v}_j$\\
    \textbf{Obtain:} FAB $p$-value $p^{\mathrm{FAB}}_j$ as defined in Equation \ref{eq:FAB corr p}.
}
\textbf{Return}: $\{p^{\mathrm{FAB}}_j\}_{j=1}^p$
\end{algorithm}
\vspace{0.25in}
}

\subsection{FAB Correlation Structure Testing} \label{sec: FAB corrStruct testing}

Section \ref{sec: FAB cor testing} builds the foundation of using FAB for correlation coefficient testing. To extend the methodology to testing every element in a large correlation matrix, computational cost needs to be reduced. This section proposes a simple divide-and-conquer method that reduces computation complexity while building a bridge that connects FAB and UMPU correlation tests.

We focus on the linking model in Equation \ref{eq:FAB corr linking model}. The linking model is constructed with $\bm{z}$ the
$q(q-1)/2$ dimensional
vectorized Fisher-transformed correlation matrix $\mathbf{Z} = F(\mathbf{P})$. We sub-divide  $\bm{z}$ into several groups of sub-vectors, with tests sharing information within but not across groups. This is expressed 
mathematically as
\begin{align}
    \widehat{\bm{z}} := 
    \begin{bmatrix}
    \widehat{\bm{z}}^{(1)} \\
    \vdots\\
    \widehat{\bm{z}}^{(m)}
    \end{bmatrix}
    &\sim 
    \cN \left( 
    \begin{bmatrix}
    \bm{z}^{(1)} \\
    \vdots\\
    \bm{z}^{(m)}
    \end{bmatrix}
    ,
    \begin{bmatrix}
    \bm{\Omega_z}^{(1)} & \hdots & ? \\
    \vdots              & \ddots & ? \\
    ?                   & ?      & \bm{\Omega_z}^{(m)}
    \end{bmatrix}
    \right)  \label{eq: breakdown sampling model}\\
    \bm{z} := 
    \begin{bmatrix}
    \bm{z}^{(1)} \\
    \vdots\\
    \bm{z}^{(m)}
    \end{bmatrix}
    &\sim 
    \cN \left( 
    \begin{bmatrix}
    \mathbf{W}^{(1)}& \hdots & \bm{0} \\
    \vdots          & \ddots & \bm{0} \\
    \bm{0}          & \bm{0} & \mathbf{W}^{(m)}
    \end{bmatrix}
    \begin{bmatrix}
    \bm{\eta}^{(1)} \\
    \vdots\\
    \bm{\eta}^{(m)}
    \end{bmatrix}
    ,
    \begin{bmatrix}
    \bm{\Psi}^{(1)} & \hdots & \bm{0} \\
    \vdots          & \ddots & \bm{0} \\
    \bm{0}          & \bm{0} & \bm{\Psi}^{(m)}
    \end{bmatrix}
    \right) \label{eq: breakdown linking model}
\end{align}

The question marks in Equation \ref{eq: breakdown sampling model} represent the cross-covariance matrix between $\widehat{\bm{z}}^{(k_1)}$ and $\widehat{\bm{z}}^{(k_2)}$ when $k_1 \neq k_2$. The block-diagonal design of the linking model in Equation \ref{eq: breakdown linking model} prohibits sharing of information across groups while still being a valid linking model. The FAB $p$-value calculation, which involves large matrix manipulation, can be reduced to the following equivalent calculation  involving smaller-scale matrix calculations:
\begin{align}
    \widehat{\bm{z}}^{(k)} &\sim \cN\left( \bm{z}^{(k)}, \bm{\Omega_z}^{(k)}\right)\\
    \bm{z}^{(k)} &\sim \cN \left(\mathbf{W}^{(k)} \bm{\eta}^{(k)},  \bm{\Psi}^{(k)}\right)
\end{align}
for all $k \in \{1:m\}$. This simplification reduces dimensions of matrix manipulation for each test and therefore achieves better scalability. Compared with the full FAB correlation test without grouping and the block-diagonal structure of $\mathbf{W}$, the approach limits sharing of information to within the same test groups. This restriction does not necessarily harm power, as the grouping-enabled more granular linking model can potentially provide better flexibility and fit, which are key factors determining the power. When $m$, the total number of test groups, equals 1, this reduces to the original FAB correlation test, whereas when $m = p$, there is equivalence to the UMPU correlation test. Thus, this divide-and-conquer approach establishes a granularity spectrum from the vanilla UMPU to the full FAB correlation test without changing the number of individual hypotheses.

Each one of the $p$ tests is assigned into one group vector $\widehat{\bm{z}}^{(k)}$ equipped with a model parameterized by $\mathbf{W}^{(k)}$ and $\bm{\eta}^{(k)}$. The design of $\mathbf{W}^{(k)}$ and the grouping of tests are arbitrary. For the design of $\mathbf{W}^{(k)}$, it can be as simple as a $p_{(k)} \times 1$ dimensional matrix of ones, if an external dataset is not available, or a matrix with columns of any power of the external data $\widehat{\bm{z}}^{(k)}_{\text{ext}}$, if external data are available. The prior represents a single mean model 
\[
\forall j \in \{1, \cdots, p_{(k)}\}, \quad z_j^{(k)} \sim \cN(\eta^{(k)}, \psi_{(k)}^2)
\]
and the latter simulates a multivariate linear regression between the true parameter $\bm{z}^{(k)}$ and the parameter's external observations $\widehat{\bm{z}}^{(k)}_{\text{ext}}$
\[
\bm{z}^{(k)} \sim \cN \left(  \sum_{d=0}^D ({\widehat{\bm{z}}^{(k)}}_{\text{ext}})^d \cdot \eta_d , \bm{\Psi}^{(k)} \right)
\]

Hierarchical structure can also be added to form Bayesian shrinkage on $\bm{\eta}$. The design is highly flexible. For grouping, a convenient yet useful approach would be to assign Fisher-transformed correlation coefficients $\widehat{z}_j$ that are likely to share similar values into the same group, based exclusively on estimation from external knowledge. For example, consider 9 test statistics generated from the testing dataset $\{\widehat{z}_1, \widehat{z}_2, \widehat{z}_3, \widehat{z}_4, \widehat{z}_5, \widehat{z}_6, \widehat{z}_7, \widehat{z}_8, \widehat{z}_9\}$ and their corresponding statistics generated from the external auxiliary dataset $\{\widehat{z}^{\text{ext}}_j\}_{j=1:9}$. We can rank $\{\widehat{z}^{\text{ext}}_j\}_{j=1:9}$ by their magnitude. Suppose the resulting order is $\widehat{z}^{\text{ext}}_5> \widehat{z}^{\text{ext}}_4>\widehat{z}^{\text{ext}}_8> \widehat{z}^{\text{ext}}_1> \widehat{z}^{\text{ext}}_3> \widehat{z}^{\text{ext}}_9> \widehat{z}^{\text{ext}}_6> \widehat{z}^{\text{ext}}_7> \widehat{z}^{\text{ext}}_2$. Then, suppose $m=3$, the final group assignment is $\{\widehat{z}_5, \widehat{z}_4, \widehat{z}_8\}_{(1)}$, $\{\widehat{z}_1, \widehat{z}_3, \widehat{z}_9\}_{(2)}$, and $\{\widehat{z}_6, \widehat{z}_7, \widehat{z}_2\}_{(3)}$. We illustrate such an assignment mechanism in Figure \ref{fig: DnQ}. This is similar to the idea of developing local regression models between $\bm{z}$ and $\widehat{\bm{z}}^{\text{ext}}$, which is helpful to address the global non-linearity between $\bm{z}$ and $\widehat{\bm{z}}^{\text{ext}}$ by offering local regression between $\widehat{\bm{z}}^{(k)}$ and $\widehat{\bm{z}}^{(k)}_{\text{ext}}$ for all $k \in \{1:m\}$. 

\begin{figure}[H]
    \centering
    \includegraphics[width = 0.6\linewidth]{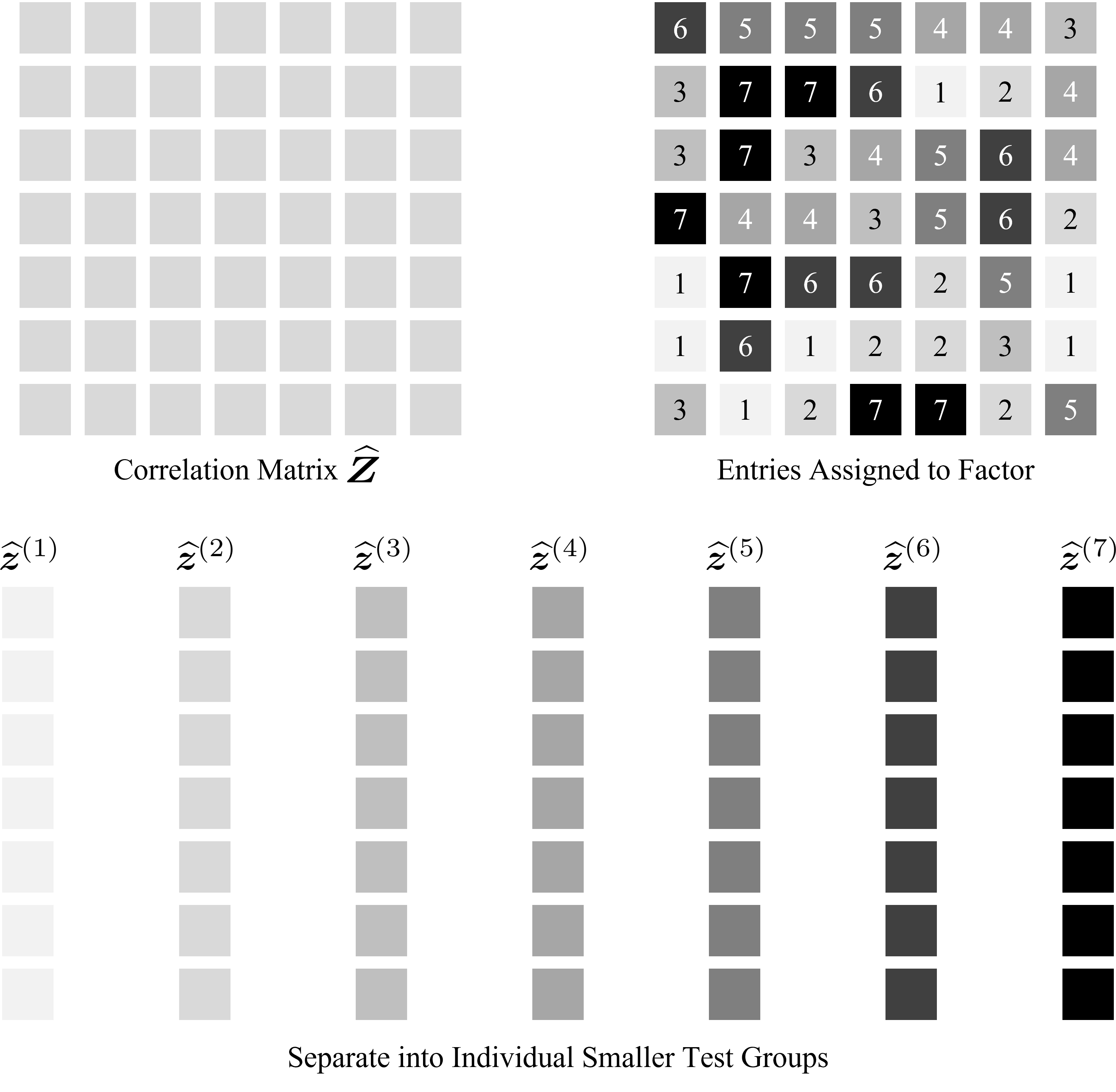}
    \caption{Divide and Conquer Group Assignment}
    \label{fig: DnQ}
\end{figure}
The most convenient way of constructing the linking model is to design $\mathbf{W}^{(k)}$ as a $p_{(k)} \times 1$ dimensional matrix of ones. This design indicates that each scalar correlation coefficient from $\bm{z}^{(k)}$ originates from the same component and shares the same scalar mean $\eta^{(k)}$. Mathematically, this can be expressed as:
\begin{align}
    \bm{z} &\sim \cN \left( \mathbb{I}\bm{\eta}, \bm{\psi}^2 \mathbf{I}_p\right)
    \label{eq: Useful Linking Model}
    \\
    \mathbb{I} = 
    \begin{bmatrix}
    \bm{1}^{(1)}_{p_{(1)}\times 1}    & \hdots & \bm{0} \\
    \vdots          & \ddots & \bm{0} \\
    \bm{0}          & \bm{0} & \bm{1}^{(m)}_{p_{(m)}\times 1}
    \end{bmatrix} 
    \quad
    \bm{\eta} &= 
    \begin{bmatrix}
    \eta^{(1)} \\
    \vdots\\
    \eta^{(m)}
    \end{bmatrix}
    \quad
    \bm{\psi}^2 \mathbf{I}_p
    = 
    \begin{bmatrix}
    \psi^2_{(1)}\mathbf{I}_{p_{(1)}}    & \hdots & \bm{0} \\
    \vdots                          & \ddots & \bm{0} \\
    \bm{0}                          & \bm{0} & \psi^2_{(m)}\mathbf{I}_{p_{(m)}}
    \end{bmatrix} \nonumber  
\end{align}

The design in Equation \ref{eq: Useful Linking Model} will be used in Section \ref{sec: simulation} for simulation. The way to perform the FAB correlation test in one group is identical as in Algorithm \ref{alg: FAB correlation test}. To perform the FAB test for all correlation groups, simply apply algorithm \ref{sec: simulation} $m$ times iteratively for each group.

\section{Simulation} \label{sec: simulation}

In this section, we apply simulation studies to demonstrate the effectiveness of the FAB correlation structure test. Section \ref{sec: purely external simulation} evaluates the methods of Section \ref{sec: independent case}, where we assume a similar external dataset is available to provide indirect information for testing. Section \ref{sec: bootstrap simulation} evaluates the methods of Section \ref{sec: nonindependent case}, where we only assume the availability of prior information on grouping.

The simulated data generating process is identical for both subsections. We generated positive semi-definite covariance matrix $\bm{\Sigma}$ by letting
    $\bm{\Sigma} = \mathbf{U}^\top\mathbf{U} + \mathbf{I}_q,$ where 
$\mathbf{U}$ is a randomly generated $l \times q$ dimensional matrix with $l< q$, $l = 50$ and $q = 100$. Each row of $\mathbf{U}$ is randomly masked with a proportion of zeros. The proportion is modulated until $\bm{\Sigma}$ has a similar number of zero and non-zero entries. This randomly generated covariance matrix $\bm{\Sigma}$, accompanied with a random mean vector, is used to generate the observable data $\mathbf{X}$. 

\subsection{\texorpdfstring{$\widehat{\bm{z}}_j^{\prime} = \widehat{\bm{z}}_{-j}^{\text{ext}}$}{Lg}: Sampling Model Using External Data} \label{sec: purely external simulation}


We first explore a specific case where $n = 100$, $q = 100$, and a similar external dataset also has $n_{\text{ext}} = 100$. This external dataset $\mathbf{X}_{\text{ext}}$ is generated from the same $\bm{\Sigma}$ as the testing dataset but with an additional random noise. To speed up computation, we applied the divide-and-conquer approach of Section \ref{sec: FAB corrStruct testing}. We manually set each group to contain 50 hypotheses, thus resulting into $m=99$ hypothesis groups. The group assignment is based on their external data magnitude order as described at the end of Section \ref{sec: FAB corrStruct testing}. Applying the linking model from Equation \ref{eq: Useful Linking Model}, the full model for the $j^{th}$ test within the $k^{th}$ test group is
\begin{align*}
    \begin{bmatrix}
        \widehat{z}_j \\
        (\widehat{\bm{z}}_{-j}^{\text{ext}})^{(k)}
    \end{bmatrix}
    &\sim \cN \left( 
    \begin{bmatrix}
        z_j\\
        ( \bm{z}_{-j}^{\text{ext}} )^{(k)}
    \end{bmatrix}
    , \frac{1}{n-3} 
    \begin{bmatrix}
        1 & \bm{0}\\
        \bm{0} & {\bm{\Omega}}_{(\bm{z}_{-j}^{\text{ext}})^{(k)}}
    \end{bmatrix}
    \right) \\
    \bm{z} &\sim \cN(\bm{1}_{p_{(k)}\times 1}\eta^{(k)}, \psi^2_{(k)} \mathbf{I}_{p_{(k)}}   ) \label{eq:FAB corr linking model}
\end{align*}

The simulation generated 2231 true null hypotheses and 2719 alternative hypotheses. The same hypothesis appears twice in the upper and lower triangle part of the covariance matrix $\bm{\Sigma}$. In this case, we count them as one hypothesis. It appears that the borrowing of information in our FAB approach leads to some UMPU $p$-values being pushed down below the pre-specified threshold for significance. This phenomenon is illustrated in Figure \ref{fig: Idpt UMPU FAB scatter}.  



\begin{figure}[H]
    \begin{center}
        \begin{subfigure}[b]{0.45\textwidth}
                \centering
                \includegraphics[width=\linewidth]{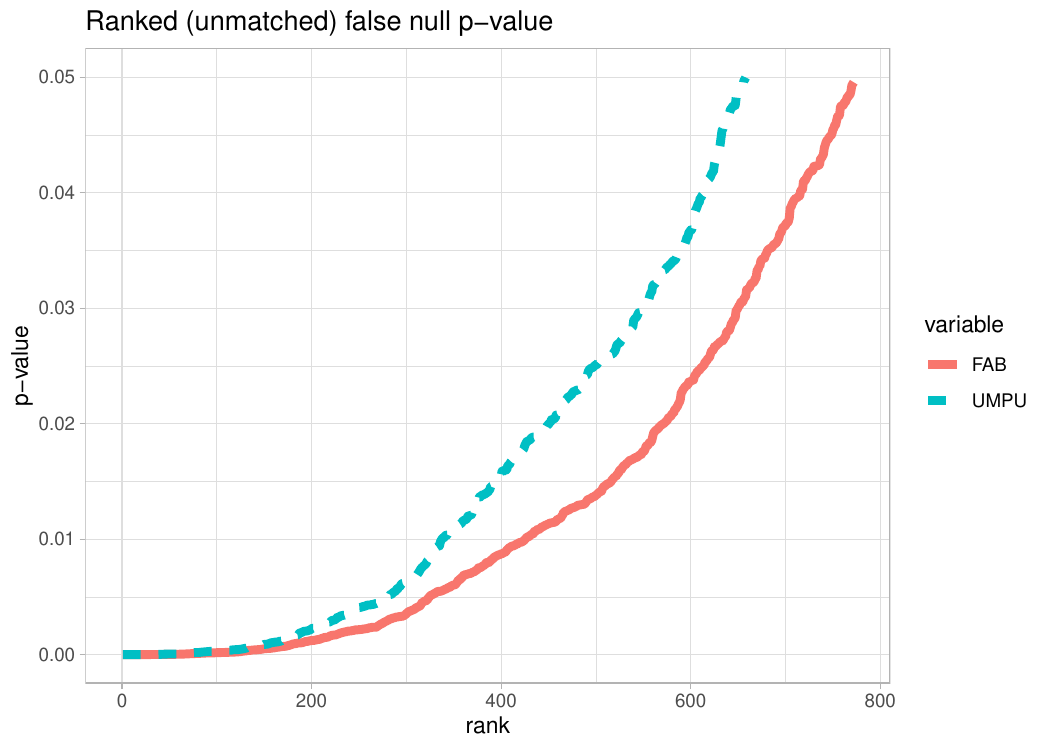}
                \caption{Ranked UMPU and FAB $p$-value}
                \label{fig: Idpt UMPU FAB power line}
        \end{subfigure}
        \begin{subfigure}[b]{0.45\textwidth}
                \centering
                \includegraphics[width=\linewidth]{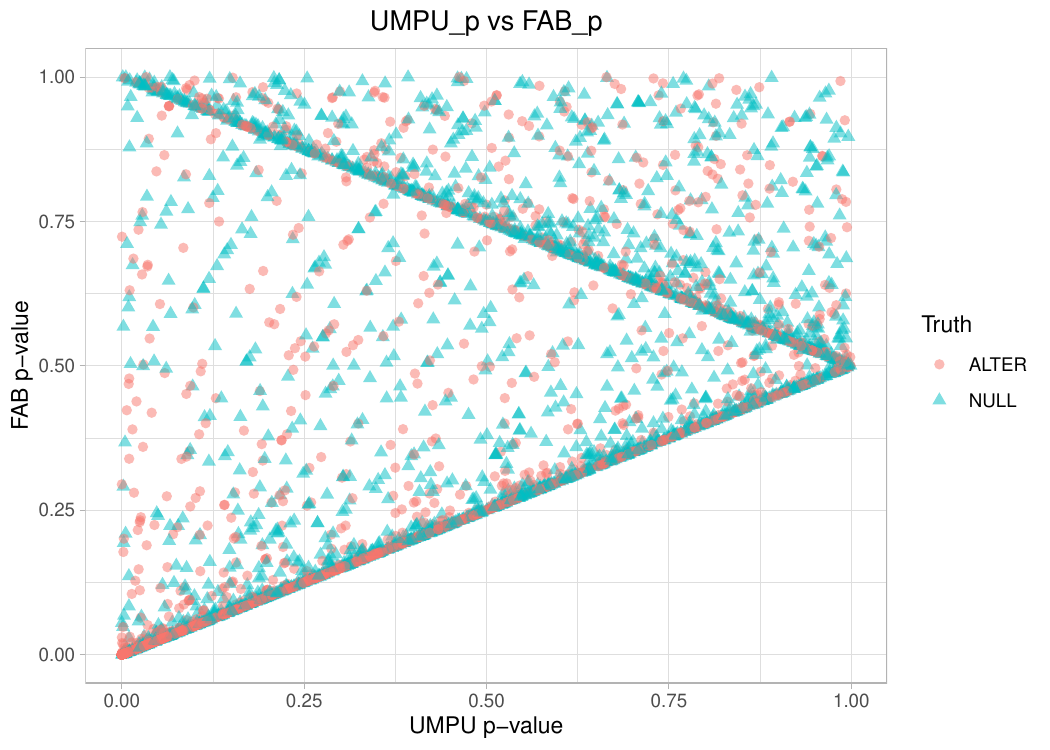}
                \caption{UMPU vs FAB $p$-value}
                \label{fig: Idpt UMPU FAB scatter}
        \end{subfigure}
    \end{center}
    \caption{UMPU vs FAB $p$-values when $\widehat{\bm{z}}_j^{\prime} = \widehat{\bm{z}}_{-j}^{\text{ext}}$. Figure \ref{fig: Idpt UMPU FAB power line} indicates ranked UMPU and FAB $p$-values. FAB has smaller $p$-values. Figure \ref{fig: Idpt UMPU FAB scatter} plots UMPU versus FAB $p$-values.}
\end{figure}

Figure \ref{fig: Idpt UMPU FAB NULL} shows the distribution of FAB $p$-values under the true null hypothesis, compared to that of UMPU, with both being uniformly distributed. Hence, the Type I error is controlled. Meanwhile, Figure \ref{fig: Idpt UMPU FAB ALTER} demonstrates the distribution of the FAB $p$-values under the alternative hypothesis is generally tilted towards the left, compared to that of UMPU, indicating a larger power and agreeing with Figure \ref{fig: Idpt UMPU FAB power line}.

\begin{figure}[H]
    \begin{center}
        \begin{subfigure}[b]{0.49\textwidth}
            \centering
            \includegraphics[width = \linewidth]{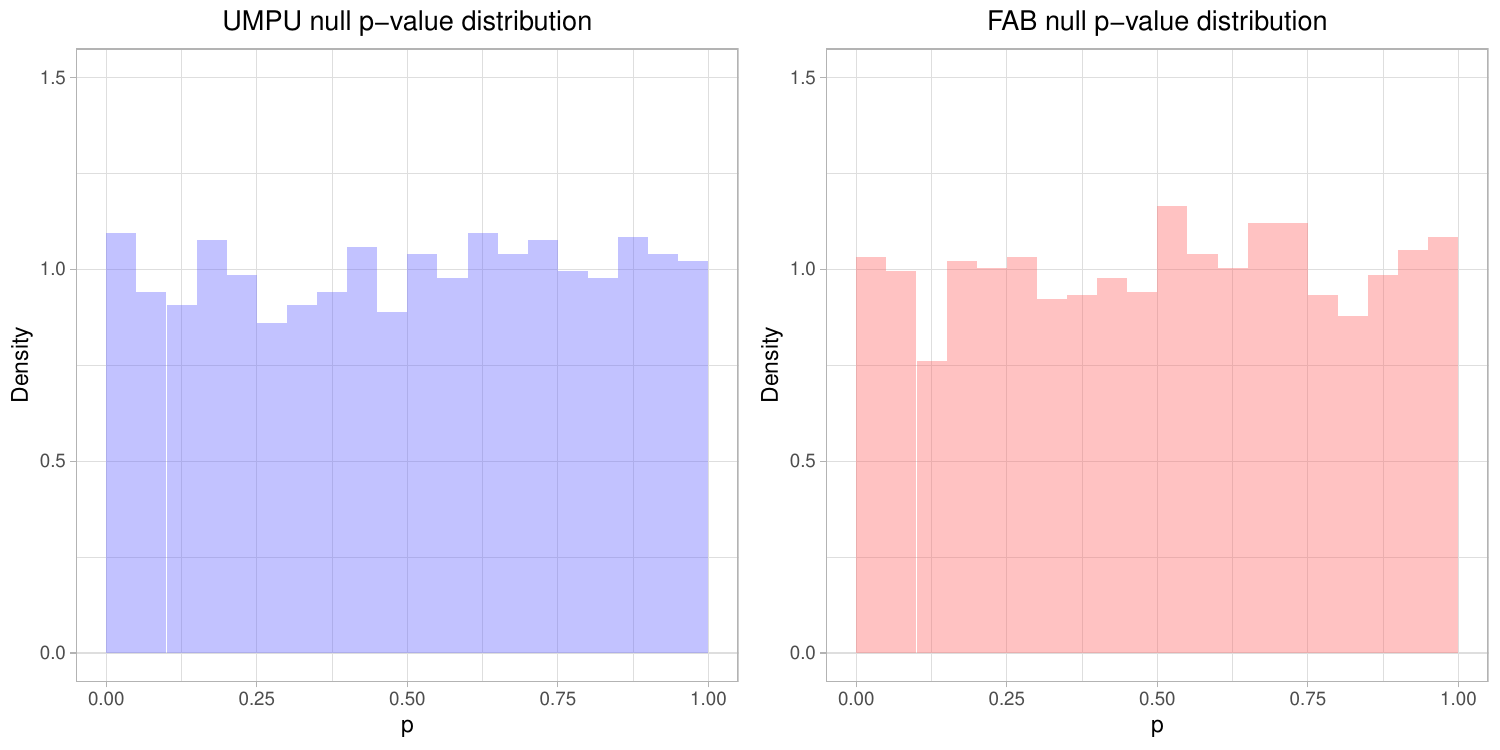}
            \caption{$p$-value distribution under true null}
            \label{fig: Idpt UMPU FAB NULL}
        \end{subfigure}
        \begin{subfigure}[b]{0.49\textwidth}
            \centering
            \includegraphics[width = \linewidth]{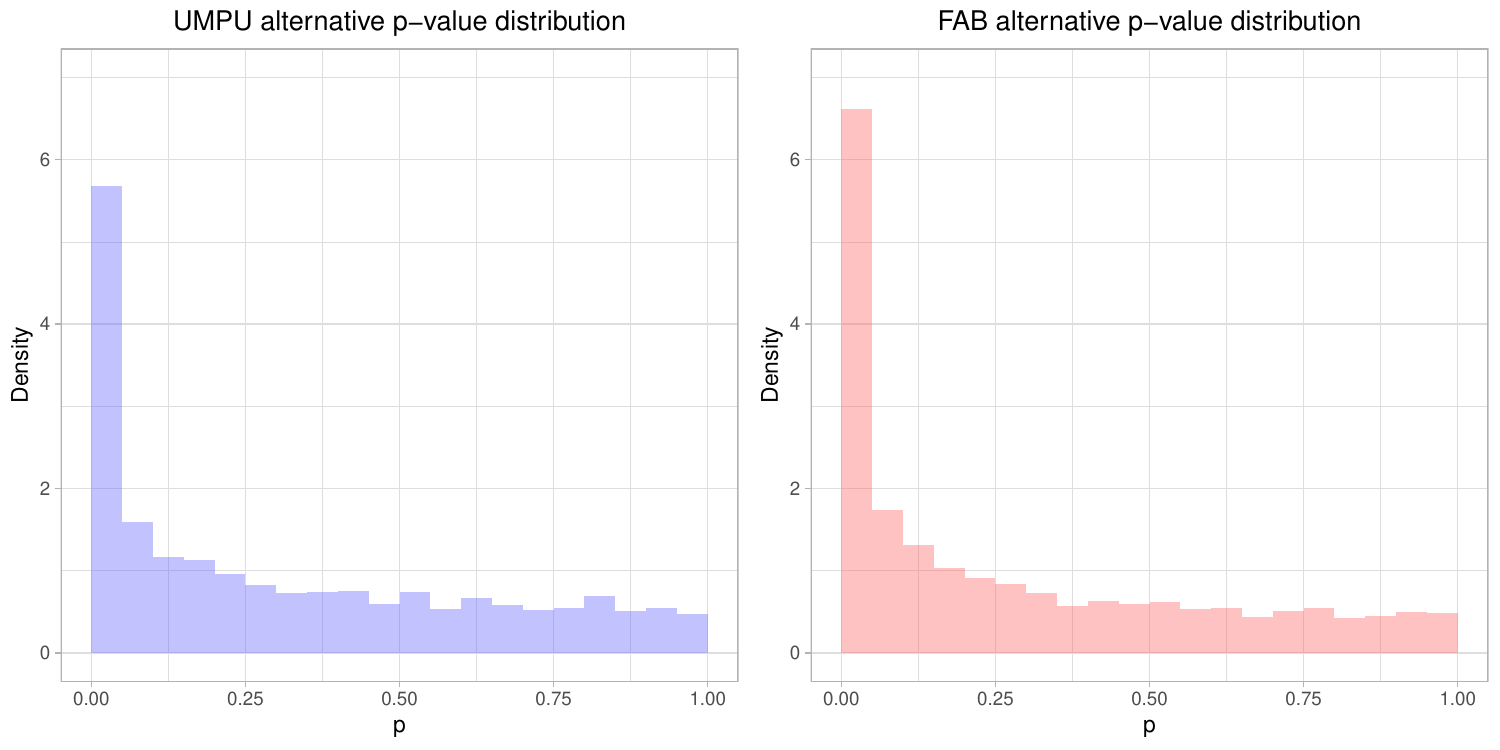}
            \caption{$p$-value distribution under alternative}
            \label{fig: Idpt UMPU FAB ALTER}
        \end{subfigure}
    \end{center}
    \caption{UMPU and FAB $p$-values distribution, when $\widehat{\bm{z}}_j^{\prime}  = \widehat{\bm{z}}_{-j}^{\text{ext}}$. Under null (\ref{fig: Idpt UMPU FAB NULL}), UMPU and FAB $p$-values are uniform. Under alternative (\ref{fig: Idpt UMPU FAB ALTER}) FAB leads to more small $p$-values.} \label{fig: Idpt UMPU FAB DIST}
\end{figure}

We generated datasets with all nine combinations of $(n, q)$ within $n \in \{50, 100, 200\}$ and $q \in \{50, 100, 200\}$.  Within each configuration, we run both FAB and UMPU correlation structure tests on $10$ different randomly generated datasets.  
The results are in Table \ref{tbl:idpt massive}. 
FAB maintains Type I error control at $p=0.05$. Under all $(n, q)$ configurations, we notice a boost in power ranging from $12.52\%-23.80\%$ comparing to UMPU. The power increment is inversely related to sample size; as sample size increases, more tests will be rejected, leaving less room for improvement for FAB.
Furthermore, the FAB offset component $b_j$ in its $p$-value expression as in Equation \ref{eq:FAB corr p} also shrinks in value when $n$ increases.

\begin{table}[H]
\centering
\resizebox{\textwidth}{!}{
\renewcommand{\arraystretch}{0.7}
\begin{tabular}{c|c|c|cllcllcll}
\Xhline{2\arrayrulewidth}
\multicolumn{3}{l|}{\multirow{3}{*}{}} &  & \multicolumn{8}{c}{$n$}                                                              \\ 
\cline{5-12}
\multicolumn{3}{l|}{}                                               &   & \multicolumn{2}{c}{50}  &  &\multicolumn{2}{c}{100} &  &\multicolumn{2}{c}{200}      \\ 
\cline{5-6}\cline{8-9}\cline{11-12}
\multicolumn{3}{l|}{}                                               &   & Null    & Alternative       &  & Null    & Alternative      &  & Null   & Alternative        \\ 
\hline
\multirow{6}{*}{$q$}  & \multirow{2}{*}{50}  & Not Reject           &   & 0.9502  & 0.8146            &  & 0.9510  & 0.7278           &  & 0.9541 & 0.5992             \\
                      &                      & Reject               &   & 0.0498  & 0.1854 (+23.77\%) &  & 0.0490  & 0.2822 (+17.39\%)&  & 0.0459 & 0.4008 (+14.06\%)  \\
\cline{3-12}
                      & \multirow{2}{*}{100} & Not Reject           &   & 0.9483  & 0.8090            &  & 0.9509  & 0.7130           &  & 0.9525 & 0.5948             \\
                      &                      & Reject               &   & 0.0517  & 0.1910 (+21.58\%) &  & 0.0490  & 0.2870 (+20.28\%)&  & 0.0475 & 0.4052 (+12.52\%)  \\
\cline{3-12}
                      & \multirow{2}{*}{200} & Not Reject           &   & 0.9488  & 0.8060            &  & 0.9499  & 0.6968           &  & 0.9486 & 0.5787             \\
                      &                      & Reject               &   & 0.0512  & 0.1940 (+23.80\%) &  & 0.0501  & 0.3032 (+18.90\%)&  & 0.0514 & 0.4213 (+13.19\%)  \\
\Xhline{2\arrayrulewidth}
\end{tabular}}
\caption{Confusion Matrix of FAB Correlation Structure Tests when $\widehat{\bm{z}}_j^{\prime}  = \widehat{\bm{z}}_{-j}^{\text{ext}}$ is external}
\label{tbl:idpt massive}
\end{table}

\subsection{ \texorpdfstring{$\widehat{\bm{z}}_j^{\prime} = \widehat{\bm{z}}_{-j}$}{Lg}: Sampling Model Using Internal Data  } \label{sec: bootstrap simulation}


All FAB $p$-values are calculated via Algorithm \ref{alg: FAB correlation test}. We set the group size to $5$ in our analyses in this section, with a 
total number of groups of $m=990$. There are 2220 true null hypotheses and 2730 alternative hypotheses. In Figure \ref{fig: boot UMPU FAB power line} we observe that FAB rejects more hypothesis.




\begin{figure}[H]
    \begin{center}
        \begin{subfigure}[b]{0.45\textwidth}
                \centering
                \includegraphics[width=\linewidth]{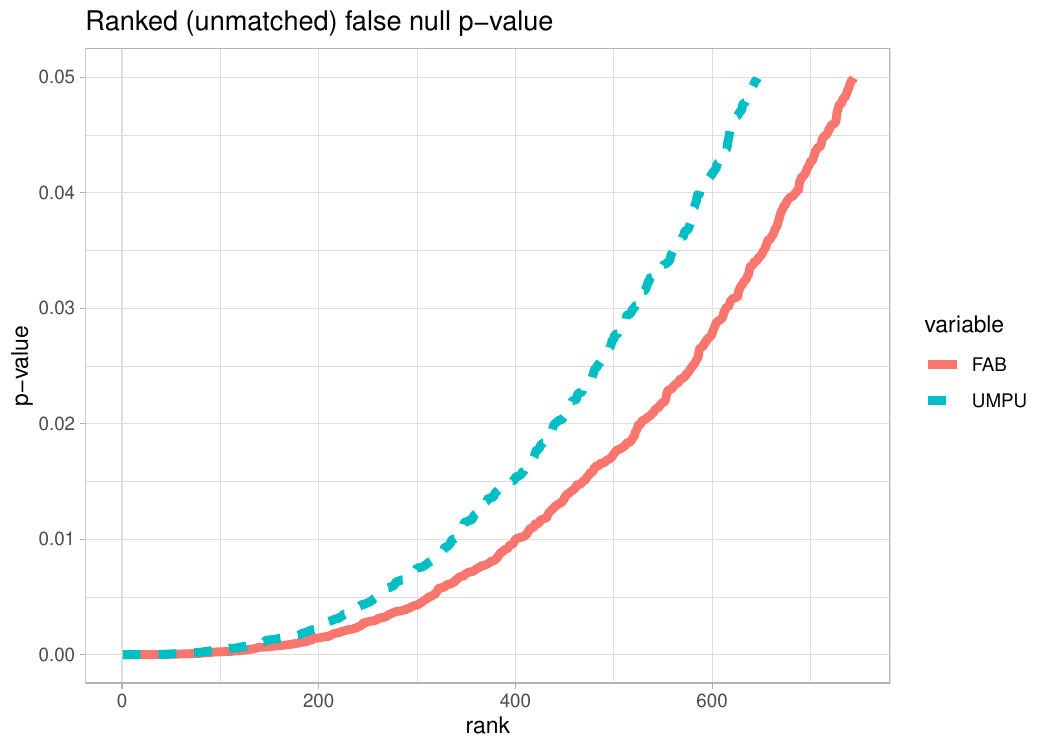}
                \caption{Ranked UMPU and FAB $p$-value}
                \label{fig: boot UMPU FAB power line}
        \end{subfigure}
        \begin{subfigure}[b]{0.45\textwidth}
                \centering
                \includegraphics[width=\linewidth]{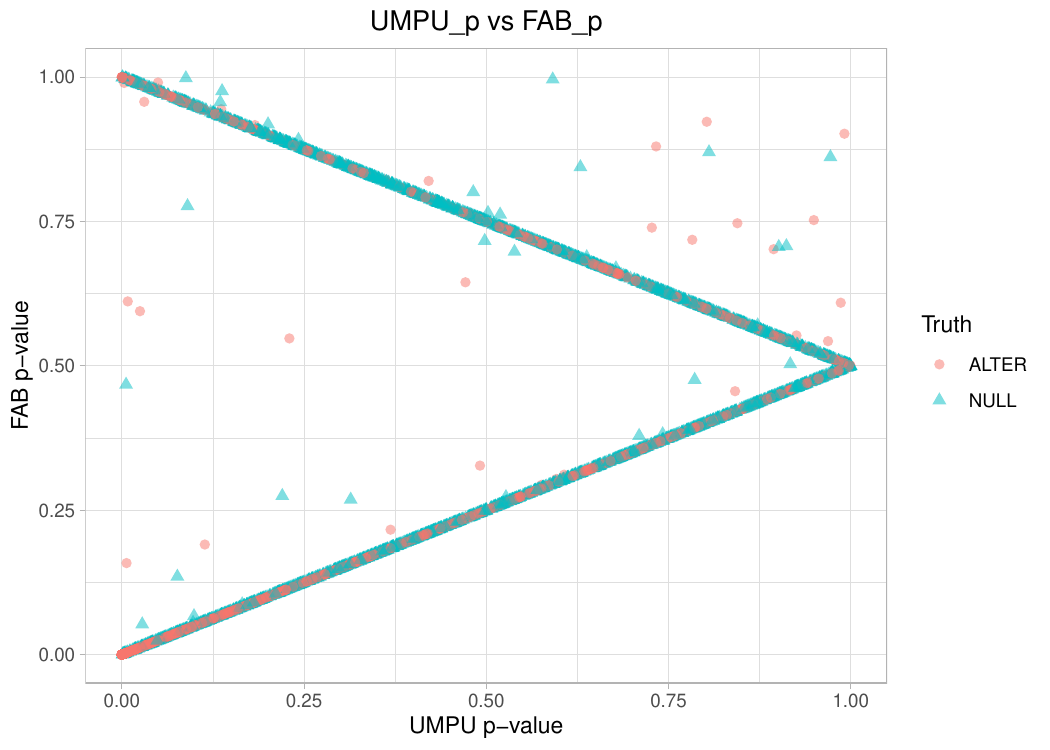}
                \caption{UMPU vs FAB $p$-values }
                \label{fig: boot UMPU FAB scatter}
        \end{subfigure}
    \end{center}
    \caption{UMPU vs FAB $p$-values when $\widehat{\bm{z}}_j^{\prime} =\widehat{\bm{z}}_{-j}$. Figure \ref{fig: Idpt UMPU FAB power line} indicates ranked UMPU and FAB $p$-values. FAB has smaller $p$-values. Figure \ref{fig: Idpt UMPU FAB scatter} indicates UMPU versus FAB $p$-values}
\end{figure}

Figure \ref{fig: boot UMPU FAB NULL} demonstrates the distribution of FAB $p$-values under the true null hypothesis, compared to that of UMPU, with both approximately uniformly distributed. Figure \ref{fig: boot UMPU FAB ALTER} shows that under the alternative hypothesis, FAB $p$-values using bootstrap are also generally smaller. The ranking of $p$-values for FAB and UMPU is demonstrated in Figure \ref{fig: boot UMPU FAB power line}.

\begin{figure}[H]
    \begin{center}
        \begin{subfigure}[b]{0.49\textwidth}
            \centering
            \includegraphics[width = \linewidth]{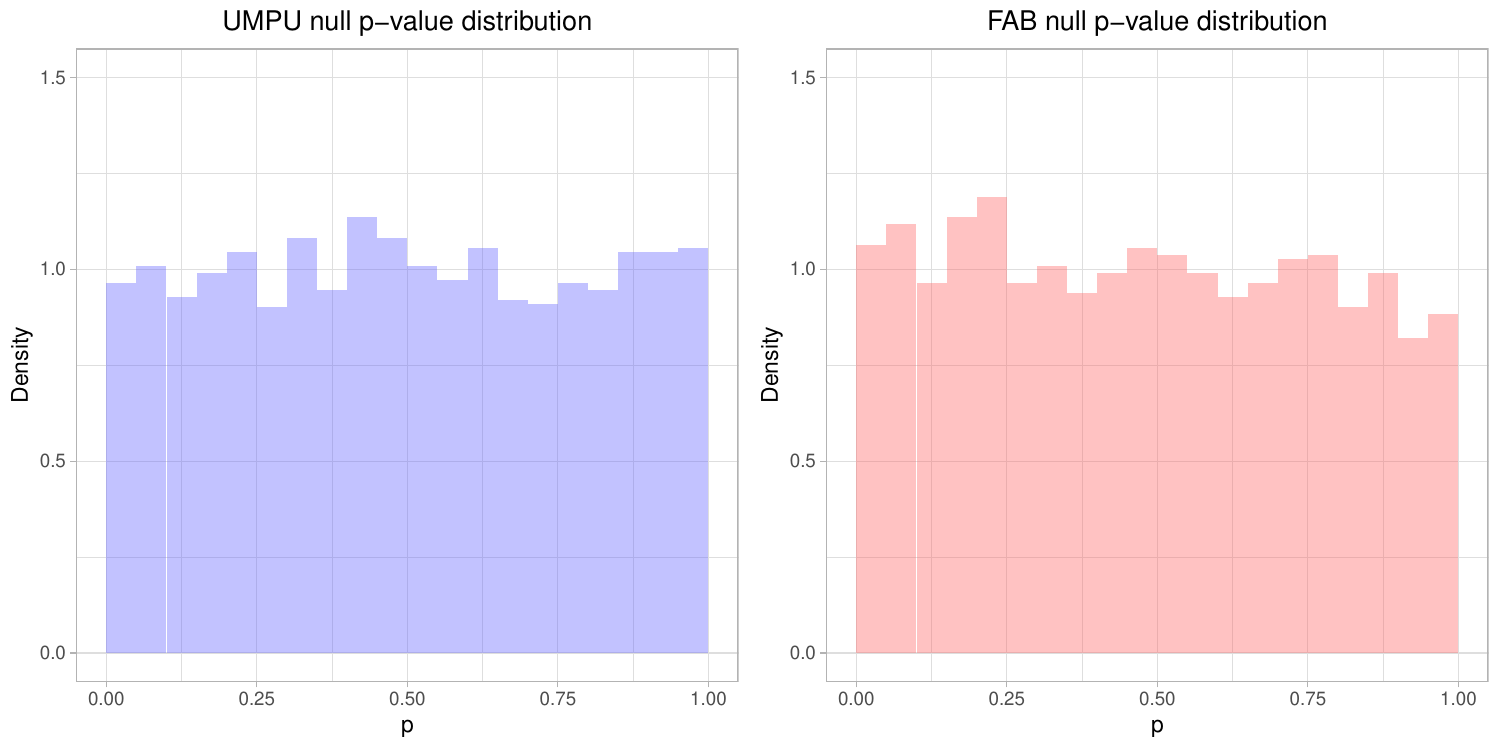}
            \caption{$p$-value distribution under true null}
            \label{fig: boot UMPU FAB NULL}
        \end{subfigure}
        \begin{subfigure}[b]{0.49\textwidth}
            \centering
            \includegraphics[width = \linewidth]{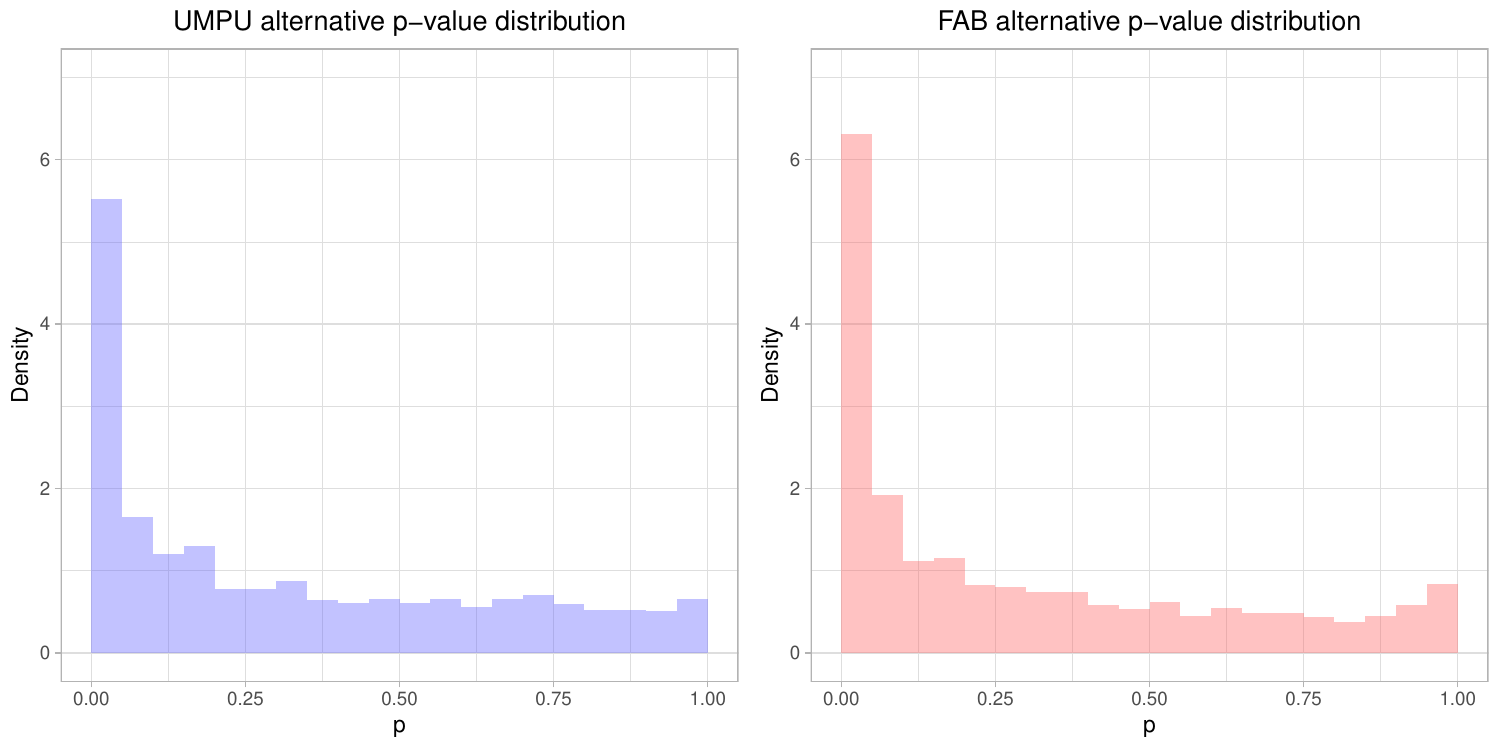}
            \caption{$p$-value distribution under hypothesis}
            \label{fig: boot UMPU FAB ALTER}
        \end{subfigure}
    \end{center}
    \caption{UMPU and FAB $p$-values distribution, when $\widehat{\bm{z}}_j^{\prime} =\widehat{\bm{z}}_{-j}$. Under null (\ref{fig: boot UMPU FAB NULL}), UMPU and FAB $p$-values are uniform. Under alternative (\ref{fig: boot UMPU FAB ALTER}) FAB leads to more small $p$-values.} \label{fig: boot UMPU FAB DIST}  
\end{figure}


\begin{table}
\centering
\resizebox{\textwidth}{!}{
\renewcommand{\arraystretch}{0.7}
\begin{tabular}{c|c|c|cllcllcll}
\Xhline{2\arrayrulewidth}
\multicolumn{3}{l|}{\multirow{3}{*}{}} &  & \multicolumn{8}{c}{$n$}                                                              \\ 
\cline{5-12}
\multicolumn{3}{l|}{}                                               &   & \multicolumn{2}{c}{50}  &  &\multicolumn{2}{c}{100} &  &\multicolumn{2}{c}{200}      \\ 
\cline{5-6}\cline{8-9}\cline{11-12}
\multicolumn{3}{l|}{}                                               &   & Null    & Alternative       &  & Null    & Alternative      &  & Null   & Alternative        \\ 
\hline
\multirow{6}{*}{$q$}  & \multirow{2}{*}{50}  & Not Reject           &  & 0.9386  & 0.7989            &  & 0.9451  & 0.7318               &  & 0.9412 & 0.6067             \\
                      &                      & Reject               &  & 0.0614  & 0.2011 (+18.64\%) &  & 0.0549  & 0.2682 (+11.61\%)    &  & 0.0588 & 0.3933 (+10.88\%)  \\
\cline{3-12}
                      & \multirow{2}{*}{100} & Not Reject           &  & 0.9475  & 0.8197            &  & 0.9464  & 0.7257               &  & 0.9494 & 0.5901             \\
                      &                      & Reject               &  & 0.0525  & 0.1803 (+16.55\%) &  & 0.0536  & 0.2743 (+14.15\%)    &  & 0.0506 & 0.4099 (+10.01\%)  \\
\cline{3-12}
                      & \multirow{2}{*}{200} & Not Reject           &  & 0.9442  & 0.8099            &  & 0.9485  & 0.7167               &  & 0.9483 & 0.5963             \\
                      &                      & Reject               &  & 0.0558  & 0.1901 (+16.84\%) &  & 0.0515  & 0.2833 (+14.19\%)    &  & 0.0517 & 0.4037 (+10.21\%)  \\
\Xhline{2\arrayrulewidth}
\end{tabular}}
\caption{Confusion Matrix of FAB vs UMPU Correlation Structure Tests using Bootstrap}
\label{tbl: boot massive}
\end{table}


The bootstrap method for the FAB test yields $10.01\%-18.64\%$ more power comparing to UMPU. Detailed results are shown in Table \ref{tbl: boot massive}. 
In Figure \ref{fig: boot consistency over NB} we verify the relationship between $n$ and $B$ and the $p$-value distribution. Under all circumstances, increasing $B$ and $n$ improves the Type I error control. 

\begin{figure}[H]
\centering
    \includegraphics[width = 0.6\linewidth]{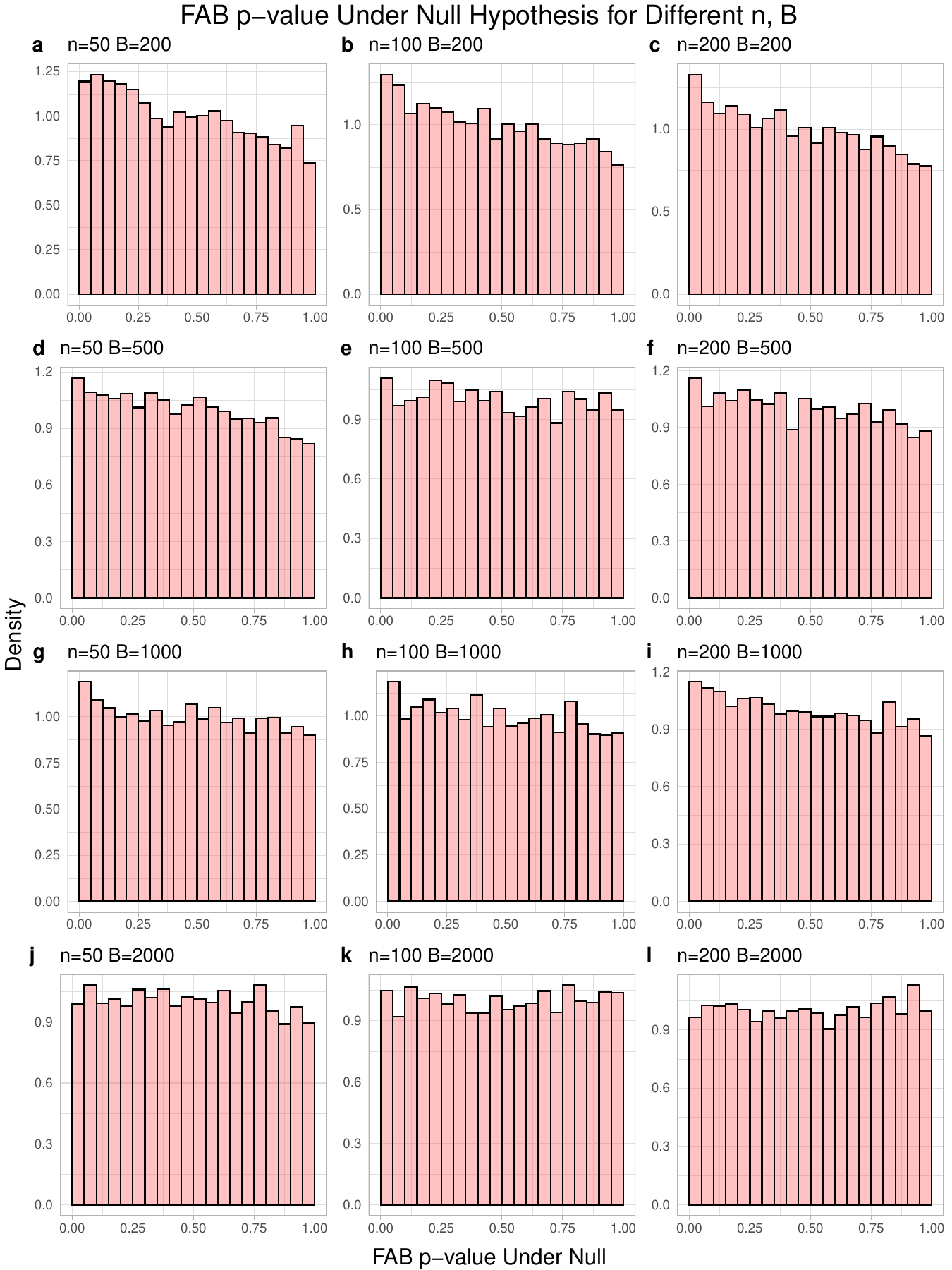}
    \caption{FAB $p$-value under null for different $(n, B)$ configuration}
    \label{fig: boot consistency over NB}
\end{figure}

\section{Application}

This section applies our FAB correlation structure tests to datasets obtained from the Cancer Dependency Map portal \citep{broad_depmap_depmap_2019}. The Cancer Dependency Map contains an extensive collection of genomics data, including measurements of gene expression, RNAi and CRISPR dependency, and drug sensitivity gathered from over 1,000 cancer cell types. The biological pathways active in these cancer cells can, in part, be described by gene-to-gene interactions, which can be inferred from the correlation structure in each genomics dataset. For instance, the correlation between genes in an expression dataset could indicate the presence of positive or negative regulation between genes. Correlation between genes in a CRISPR dependency dataset could indicate that two genes share an identical function because their deletion has similar adverse effects across cancer cell types. Because gene-to-gene interactions may be similar across cancer tissue types or specific technologies, rich genomics datasets such as those found in the Cancer Dependency Map present opportunities for sharing information using FAB structure tests.

Our first application of the FAB structure test uses a correlation matrix derived from RNAi dependency data \citep{mcfarland2018improved} to test correlation coefficients derived from CRISPR dependency data \citep{meyers2017computational}. Both datasets in our application contain 280 genes with dependency scores measured on 525 cancer cell types. Although RNAi and CRISPR are different technologies, they are both designed to measure genetic loss-of-function, so one might expect that the correlation structures uncovered by these technologies are similar, though not identical. Our second application uses a correlation matrix derived from 67 Breast cancer cell types to test the correlation structure derived from 206 Lung cancer cell types. These datasets each have 277 measured genes. Again, while these cancers and their regulatory networks are biologically distinct, one would expect some similarities between their correlation structures.

As the ground truth for the correlation structure is unknown for real data, we use as an imperfect surrogate published results for similar datasets \citep{subramanian2005gene, rouillard2016harmonizome}. We take gene pairs that are declared as significantly correlated in these findings as ground truth. Gene pairs that are not declared as correlated are not necessarily uncorrelated but can be due to a lack of power in rejecting the null hypothesis.  We use a list of 41,327 previously reported gene pairs as ground truth for those correlations that are non-zero. These pairs were curated from microarray experiments and downloaded from the Molecular Signatures Database.

Only genes appearing in the ground truth dataset will be considered in our analysis.  Focusing on these genes, and removing observations containing missing data and no variance genes, 
the RNAi and CRISPR datasets are reduced to an identical set of 69 measured genes, with the sample sizes equal to 107 and 519, respectively.
Out of the $60\times 70/2 = 2415$ pairs of genes, 165 pairs were declared as correlated in previous published results based on different datasets.
We report the empirical power as the number of rejected pairs out of the 165 divided by 165.
 Figure \ref{fig: application case 1 RNA}a plots the testing dataset's Fisher-transformed correlation coefficients $\widehat{\mathbf{Z}}$ against that of the auxiliary dataset ($\widehat{\mathbf{Z}}_\text{ext}$). We performed a simple linear regression between entries of $\widehat{\mathbf{Z}}$ and $\widehat{\mathbf{Z}}_\text{ext}$; a line through the origin provided an excellent fit, and hence was used as the linking model as in Equation \ref{eq: breakdown linking model}. Hence, $\mathbf{W}^{(m)}=\left[ \hat{z}^{\text{ext}}_{1},\ldots, {\hat{z}}^{\text{ext}}_{p_m}  \right]^\top$. 

We limit each test group size to 120 tests.  The bootstrap FAB correlation test algorithm \ref{alg: FAB correlation test} is utilized to generate $p$-values. Figure \ref{fig: application case 1 RNA}b plots the sorted $p$-values for FAB test (red) and traditional two-sided UMPU test (blue) below $\alpha = 0.05$. Clearly, FAB results in more overall rejections than UMPU. The power also increased from UMPU's 31 correct rejections to 38 correct rejections among the 165 alternative hypotheses. 

\begin{figure}[H]
    \centering
    \includegraphics[width = 0.85\linewidth]{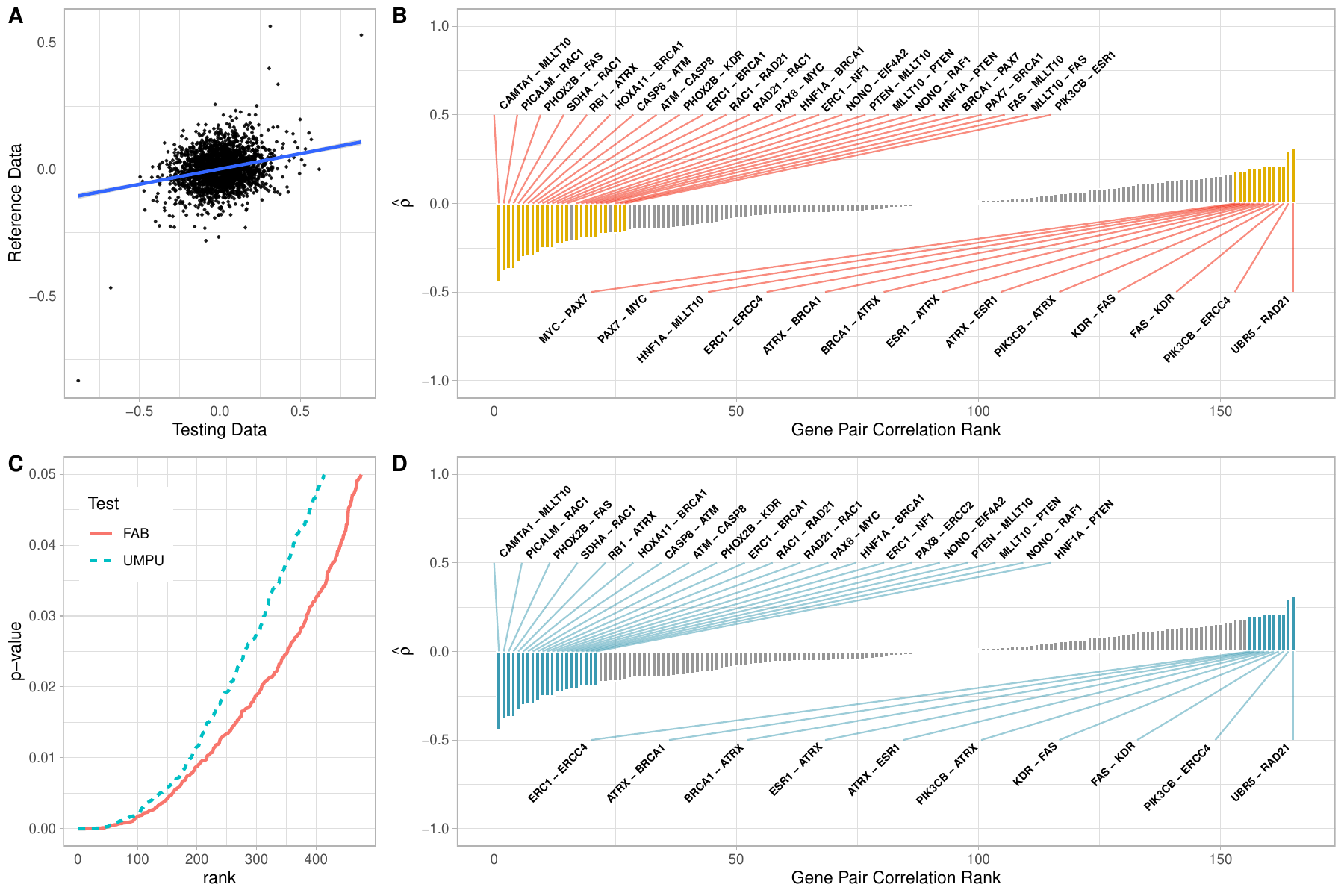}
    \caption{FAB vs UMPU under case 1: RNAi assisted by CRISPR. A: Linear relationship between lung breast correlation coefficients and that of breast cancer. C: Ranked FAB and UMPU $p$-values. FAB results in more $p$-values smaller than 0.05. B: Rejected gene pairs of FAB, D: rejected gene pairs of UMPU. FAB results more rejections than UMPU.}
    \label{fig: application case 1 RNA}
\end{figure}

We applied an identical data pre-selection strategy for the lung cancer and breast cancer datasets in case 2. After pre-selection, 
the lung cancer testing dataset and breast cancer auxiliary dataset contain 206 and 61 measurements for 67 measured genes, with 162 gene pairs declared as correlated. Figure \ref{fig: application case 2 Lung} (a) demonstrates a strong positive linear relationship between $\widehat{\mathbf{Z}}$ and $\widehat{\mathbf{Z}}_\text{ext}$, with $R^2 = 0.3463$. In the linking model, we let $\mathbf{W}^{(m)} = [ \mathbf{1}_{1:{p}_m}, \widehat{\bm{z}}_{1:{p}_m}^{\text{ext}}]$. We also applied ridge shrinkage at the level of $\lambda=0.01$ on the estimator of $\bm{\eta}^{(m)}$. Similar to the group assignment in the first case, each test group is also comprised of 120 individual tests. As indicated in Figure \ref{fig: application case 2 Lung}b, FAB rejects more hypotheses than UMPU does in general, with 96 correctly rejected hypotheses as opposed in 88 for UMPU among the 162 ``true'' alternatives. 

\begin{figure}[H]
    \centering
    \includegraphics[width = 0.85\linewidth]{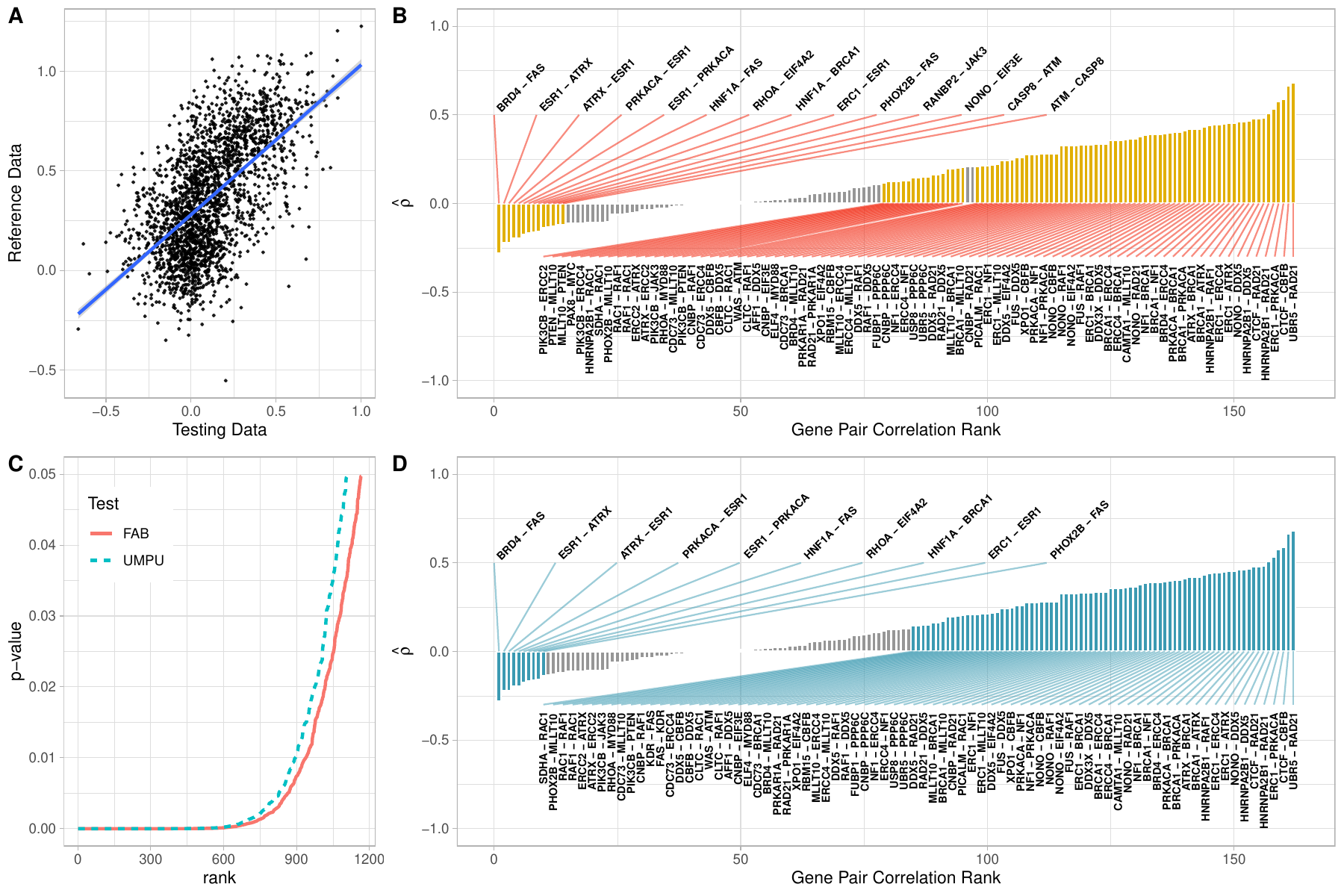}
    \caption{FAB vs UMPU under case 2: lung cancer assisted by breast cancer. A: Significant linear relationship between lung breast correlation coefficients and that of breast cancer. C: Ranked FAB and UMPU $p$-values. FAB results in more $p$-values smaller than 0.05. B: Rejected gene pairs of FAB, D: rejected gene pairs of UMPU. FAB results more rejections than UMPU.}
    \label{fig: application case 2 Lung}
\end{figure}

In addition to comparing with UMPU, we also compare FAB with the AdaPT method designed to incorporate outside information in FDR control \citep{lei2018adapt}.  In contrast to our approach, AdaPT takes a set of pre-computed independent $p$-values as input and uses available side information to help define a sequence of adaptive thresholds. The $p$-values are then compared sequentially against this sequence of thresholds to determine rejections. FAB, on the other hand, uses available side information to generate $p$-values, which are then subject to traditional p-value thresholds that ignore the side information.  As AdaPT uses side information to directly determine rejection thresholds, we anticipate that discoveries made by AdaPT may be less robust to the quality and informativeness of the external information.

To compare the three approaches, we apply the vanilla Benjamini-Hochberg (BH) \citep{benjamini1995controlling} procedure on UMPU $p$-values, vanilla BH procedure on FAB $p$-values, and take UMPU $p$-values as input for the AdaPT method. In Figure \ref{fig: FDR case 1 RNA}a and Figure \ref{fig: FDR case 2 Lung}a, the numbers of discoveries under a series of target FDR thresholds are recorded for UMPU, FAB, and AdaPT under both cases 1 and 2. FAB results in uniformly more discoveries than UMPU, but its advantage over AdaPT depends on the quality of the external dataset. In the first case, as demonstrated in Figure \ref{fig: application case 1 RNA}, $\widehat{\mathbf{Z}}_{\text{ext}}$ explains less variation ($R^2 = 0.0637$) in $\widehat{\mathbf{Z}}$ than in the second case. Thus, the outside information is less informative in the first case, and AdaPT resulted in fewer discoveries in the first case compared to the second case. In contrast, FAB had a similar advantage over the UMPU in both cases. Our conjecture is that FAB can better take advantage of noisy side information.  

\begin{figure}[H]
    \centering
    \includegraphics[width = 0.7\linewidth]{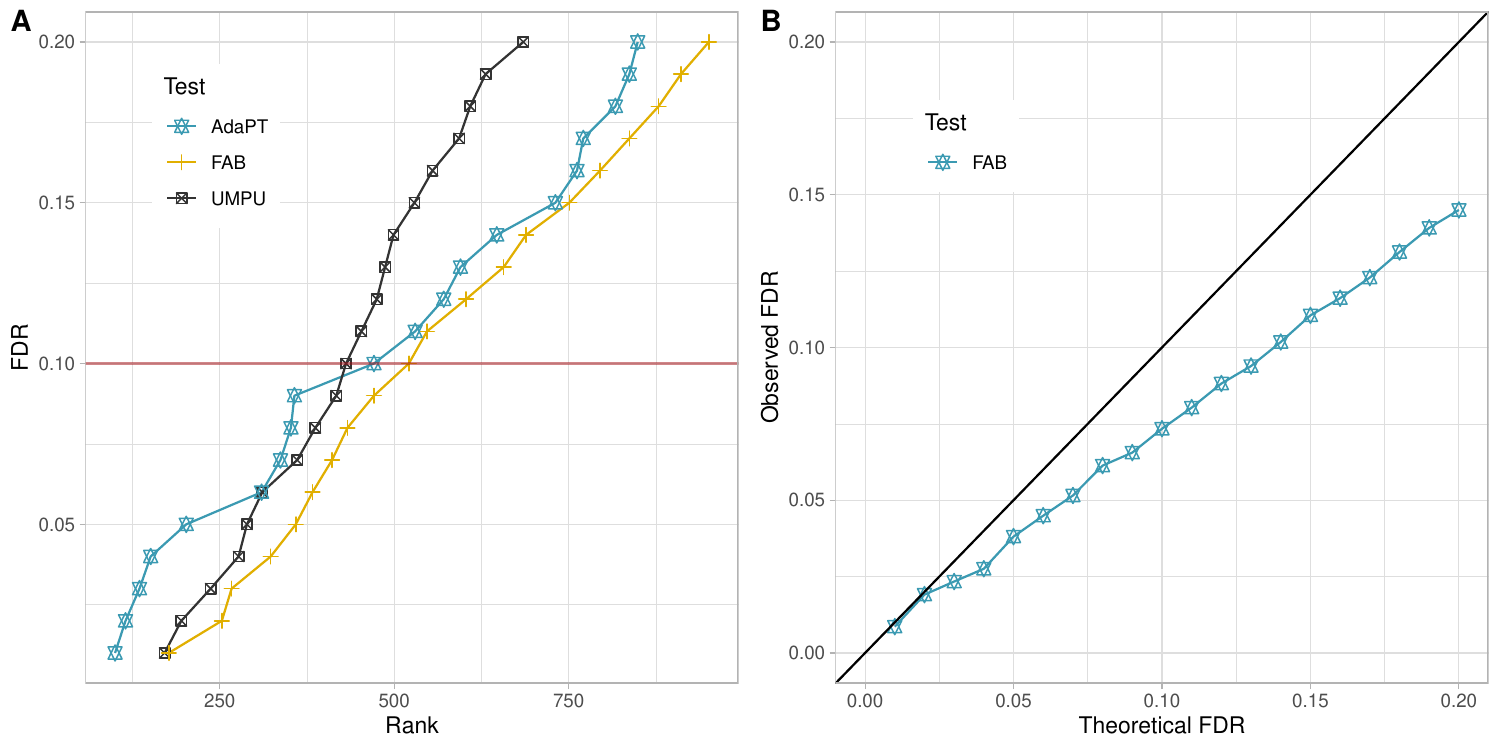}
    \caption{Discoveries comparison between UMPU, FAB, and AdaPT under FDR control: Case 1, RNAi assisted by CRISPR. A: total number of discoveries. AdaPT results in more discoveries uniformly. B: Simulated data's observed vs theoretical FDR for FAB. FAB maintains FDR control.}
    \label{fig: FDR case 1 RNA}
\end{figure}

Due to sharing of information via the linking model within test groups, FAB $p$-values tend to be positively correlated. Hence, the BH procedure may be conservative.  To investigate FDR control, we generated simulated data having identical configuration, including the sample size $m$ and data dimension $q$, with the real datasets in both cases. We repeated the same FAB analyses as used for the real data across 10 simulation replicates for each case.
In Figure \ref{fig: FDR case 1 RNA}b and Figure \ref{fig: FDR case 2 Lung}b, we recorded the observed FDR against the theoretical FDR on the simulated data for each case. FAB controls FDR at the specified levels in each case.

\begin{figure}[H]
    \centering
    \includegraphics[width = 0.7\linewidth]{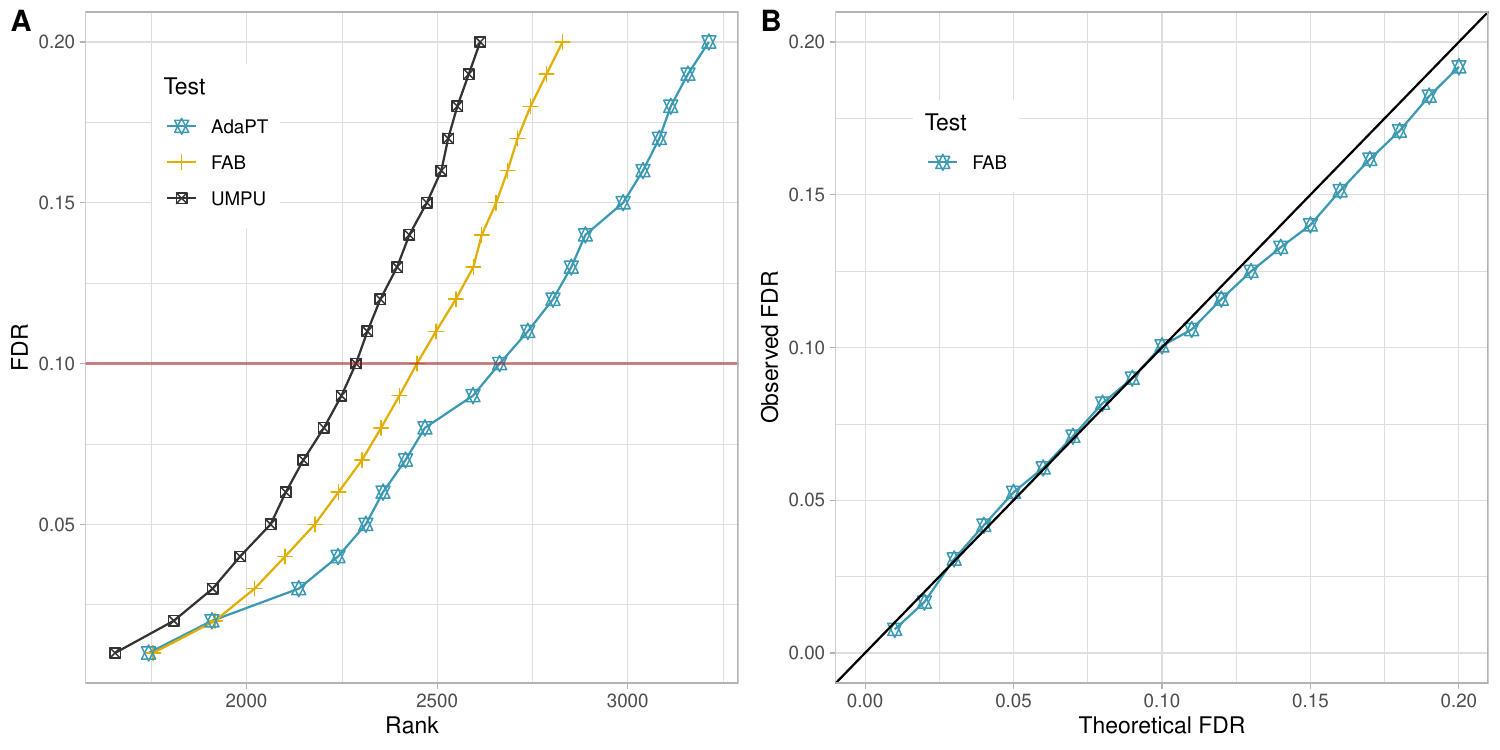}
    \caption{Discoveries comparison between UMPU, FAB, and AdaPT under FDR control: Case 2, Lung assisted by Breast. A: total number of discoveries. AdaPT results in more discoveries uniformly. B: Simulated data's observed vs theoretical FDR for FAB. FAB maintains FDR control.}
    \label{fig: FDR case 2 Lung}
\end{figure}


\section{Discussion}

This articles develops a frequentist assisted by Bayes (FAB) testing methodology for support recovery of correlation structure. Our work has demonstrated the flexibility of the FAB framework in learning from indirect information to assist hypothesis testing on correlation coefficients both when indirect information is sourced directly from external datasets or different tests within the same dataset. The simulation results have demonstrated improvement of power while maintaining asymptotic control of Type I error in both the independent and bootstrap cases. The real data application has also illustrated the methodology's capability to improve power while still offering FDR control empirically. Our results suggests that the FAB correlation structure testing framework is a ``Type I error safe'' yet flexible framework that allows customization of the predictive linking model to improve power.

Our methodology was not explicitly developed with a focus on controlling FDR and multiple testing error.  Instead we have focused on developing an approach to include outside information in obtaining $p$-values that will lead to increased power in correlation structure testing if the outside information was indeed informative.  One can then use these $p$-values as inputs into any appropriate FDR or other multiple testing method.  For simplicity we focused on Benjamini-Hochberg in our illustration.  However, given the correlation in the $p$-values it may be more appropriate to use alternative procedures that accommodate dependence among the $p$-values \citep{efron2007correlation, romano2008control, sun2009large, fan2012estimating, Ramdas2019}

\if1\blind
{
    \section{Acknowledgement}
    The authors thank Jordan Bryan for helpful conversations. This work was partially supported by National Institute of Health grants R01-ES028804 and R01-MH118927. 
} \fi

\bigskip

\begin{center}
{\large\bf SUPPLEMENTARY MATERIAL}
\end{center}

Proof of Theorem \ref{thm: consistency control} and \ref{thm:bootstrap consistency of big matrix} and the link to the Github repository, including reproducible code, reports, and knitted html files are included in the supplementary material.

\bibliography{FABcorr}
\end{document}